\documentclass[11pt,a4paper]{article}    
\usepackage{amsmath,amsfonts,amssymb,amsthm,amscd}
\usepackage[english]{babel}

\newtheorem{theorem}{Theorem}
\newtheorem{corollary}{Corollary}
\newtheorem{proposition}{Proposition}
\newtheorem{lemma}{Lemma}

\newcommand{\dsp}{\displaystyle}

\newcommand{\R}{\mathbb{R}}

\numberwithin{equation}{section}

\begin{document}
\title{On the magnetic shield for a Vlasov-Poisson plasma}
\author{Silvia Caprino$^1$, Guido Cavallaro$^2$ and Carlo Marchioro$^3$}

\maketitle

{\footnote{$^1$Dipartimento di Matematica, Universit\`a di Tor Vergata, 
Via della Ricerca Scientifica 1, 00133 Roma, Italy. E-mail:  caprino@mat.uniroma2.it}}

{\footnote{$^{2}$Dipartimento di Matematica, Sapienza Universit\`a di Roma, 
Piazzale Aldo Moro 2, 00185 Roma, Italy. E-mail:  cavallar@mat.uniroma1.it}}

{\footnote{$^{3}$International  Research Center M$\&$MOCS (Mathematics and Mechanics of Complex Systems). 
E-mail:  marchior@mat.uniroma1.it}}

\vspace{-50pt}

\begin{abstract}

We study  the screening of a bounded body $\Gamma$  against the effect of a wind of charged particles, by means
of a shield produced by a magnetic field which becomes infinite on the border of $\Gamma$. 
The charged wind is modeled by a Vlasov-Poisson plasma, the bounded body by a torus, and the external magnetic field is taken close to the border of $\Gamma$.

We study two models: a plasma composed by different species with positive or negative charges, and finite total mass of each species, and another made of many species of the same sign, each having infinite mass. We investigate the time evolution of both systems, showing in particular that the plasma particles cannot reach the body.  

Finally we discuss possible extensions to more general initial data. We show also that when the magnetic lines are straight lines, (that imposes an unbounded body), the previous results can be improved.

\end{abstract}

\medskip

\noindent Mathematics Subject Classification: {\small{Primary}}: 82D10, 35Q99, 76X05; \\
{\small{Secondary}}: 35L60.

\medskip
\noindent  Keywords: 
Magnetic shield, Vlasov-Poisson  equation, infinitely  extended  plasma, infinite velocities.

\section{Introduction and main results}
\medskip
It is well known that a magnetic field can be a shield against the effect of solar wind and cosmic rays. The most famous effect happens in the earth, in which its magnetic field pushes the cosmic rays towards the poles, causing the aurora borealis.
Analogously, one can pose the problem of how to protect a spaceship
flying to Mars from the solar wind. Indeed a prolonged exposure to the solar
wind can be dangerous for the health of the crew of the spaceship.

\noindent It is easy to show that, in the motion of a point charged particle,   surfaces with a magnetic field
singular enough represent  an impenetrable wall producing the so-called magnetic shield, as it can be seen in the proof of estimate (\ref{B}). 
In the present paper we will discuss whether a similar effect could happen for a charged plasma described by the Vlasov-Poisson equation, which evolves in presence of an external magnetic field singular on the border of some regions. Actually the problem is not obvious, since in the continuum case a fluid particle  could in principle be pushed 
against  the wall by the other particles, especially if the total mass of the plasma is infinite.
We will show rigorously that in many cases this phenomenon cannot happen. The construction
of such a model of  magnetic shield depends mainly on three different things: the geometry of the region that we want to screen by a magnetic field and the spatial as well as velocity distribution of the plasma. In the sequel we discuss some emblematic important cases. In other situations we are not able to obtain a proof of the existence of the shield effect, and we do not know if this is due to a technical lack or to physical reasons.

\bigskip

The solar wind is composed by different species of particles, i.e. electrons, protons, and other positive heavy ions,
 therefore we initially  discuss the time evolution of a plasma with many species with different sign charges, moving under the influence of the auto-induced electric field and a fixed external magnetic field. Then we study a non-trivial generalization of the initial data when only one species is present. 

Let $\Gamma$ be a smooth region of the space and $n$ the total number of the species. For any $i=1,2,...,n$ let $f_i(x,v,t)$ represent  the distribution function of charged particles at the point of the phase space $(x,v)$ at time $t$ and let  $\sigma_i$ (which can be positive or negative) be the charge per unit mass of the $i$-th species.
 We describe the time evolution of this system via the
$n$ Vlasov-Poisson equations:
\begin{equation}
\label{Eq}
\left\{
\begin{aligned}
&\partial_t f_i(x,v,t) +v\cdot \nabla_x f_i(x,v,t)+  \sigma_i(E(x,t)+v \wedge B(x)) \cdot \nabla_v f_i(x,v,t)=0  \\
&E(x,t)=\sum_{i=1}^n \sigma_i \int_{\R^3 \setminus \Gamma} \frac{x-y}{|x-y|^3} \ \rho_i(y,t) \, dy     \\
&\rho_i(x,t)=\int_{\R^3} f_i(x,v,t) \, dv \\
&f_i(x,v,0)=f_{i,0}(x,v)\geq 0,  \qquad  x\in {\R^3 \setminus \Gamma},  \qquad v\in\R^3,  \qquad i=1 \dots n
\end{aligned}
\right.
\end{equation} 
where  $\rho_i$ are the spatial densities of the species, $E$ is the electric field and $B(x)$ is an external magnetic field that is singular on the boundary $\partial \Gamma$.

\noindent The model described by these equations neglects the relativistic effects and this approximation is reasonable, being the velocities of the main part of the solar wind of the order  $4 \cdot 10^2$   Km/sec,  while the velocity of light
 is of the order   $3 \cdot 10^5$   Km/sec.

\bigskip

System (\ref{Eq}) shows that $f_i$ are time-invariant along the solutions of the so called characteristics equations:
\begin{equation}
\label{ch} 
\begin{cases}
\dsp  \dot{X}_i(t)= V_i(t)\\
\dsp  \dot{V}_i(t)= \sigma_i \left[ E(X(t),t)+V_i(t) \wedge B(X(t))\right] \\
\dsp (X_i(t'), V_i(t'))=(x,v)  \\
\dsp f_i(X_i(t), V_i(t), t) = f_{i,0}(x,v),
 \end{cases}
\end{equation}
where we have used the simplified notation
\begin{equation}
 \label{2.8}
(X_i(t),V_i(t))= (X_i(t,t',x,v),V_i(t,t',x,v)) 
 \end{equation}
 to represent a characteristic at time $t$ passing at time $t'<t$ through the point $(x,v)$. Hence we have
 \begin{equation}
 \label{2.9}
\| f_i(t)\|_{L^\infty}= \| f_{i,0} \|_{L^\infty}.
 \end{equation}
 Moreover this dynamical system preserves the measure of the phase space (Liouville's theorem).

\bigskip

We start the investigation when $\Gamma$ is a bounded region (which represents the spaceship).  For concreteness we assume $\Gamma$ to be a torus, but the study can be extended without difficulty to other smooth regions. We have chosen a torus and not, e.g., a sphere because we want to put on the border a magnetic field well defined everywhere (and for a sphere it is not defined at the poles), fact that imposes some constraints on the topology of $\Gamma$.

\bigskip

 Let $x=(x_1,x_2,x_3)\in \R^3$, then $x\in \Gamma$ if
 
 $$
\left(R-\sqrt{x_1^2+x_2^2} \right)^2+x_3^2=r^2,
$$
 for any $r\in [0,r_0]$ with $R>r_0>0$.

\bigskip

\noindent  In toroidal coordinates the equations are

\begin{equation}
\label{coord.}
\begin{cases}
\dsp  x_1= (R+r\cos \alpha) \cos \theta\\
\dsp x_2 =  (R+r\cos \alpha) \sin \theta \\
\dsp x_3 = r \sin \alpha \\
\dsp 0 \leq \alpha < 2 \pi, \quad  0 \leq \theta < 2 \pi.
 \end{cases}
\end{equation}

\noindent We will describe the region $\R^3 \setminus \Gamma$ with toroidal coordinates when $r_0<r<(R+r_0)/2$
(near the external border of the torus),
 with a smooth switch to cartesian coordinates elsewhere.

We choose the external magnetic field $B(x)$ of the form

\begin{equation}
\label{m.f.}
B(x) = \nabla \wedge A(x), \qquad  A(x)= \frac{a(r)}{R+r \cos \alpha} \hat{e}_\theta
\end{equation}
where $a(r)$ is a decreasing smooth function for $r_0 < r \leq r_0+\frac{R-r_0}{4}$, divergent for $r \to r_0^+$  and  vanishing otherwise, and  $\hat{e}_\theta$ is the unit vector tangent to the border of the torus in the direction of increasing $\theta$ (and fixed $\alpha$).

From (\ref{coord.}) one obtains

\begin{equation}
\label{versors}
\begin{cases}
\dsp   \hat{e}_r= \cos \alpha \cos \theta  \ \hat{c}_1+\cos \alpha \sin \theta \ \hat{c}_2+\sin \alpha \ \hat{c}_3\\
\dsp  \hat{e}_\theta=  - \sin \theta  \ \hat{c}_1 + \cos \theta \ \hat{c}_2\\
\dsp  \hat{e}_\alpha = -\sin \alpha \cos \theta \ \hat{c}_1-\sin \alpha \sin \theta \ \hat{c}_2+\cos \alpha  \ \hat{c}_3 
 \end{cases}
\end{equation}
where $\hat{e}_\alpha $ is the unit vector defined analogously to $\hat{e}_\theta$ (and so orthogonal to $\hat{e}_\theta$), $\hat{e}_r=\hat{e}_\theta \wedge \hat{e}_\alpha $, and $\hat{c}_1,\hat{c}_2,\hat{c}_3$ are the units vectors of the Cartesian axes $x_1,x_2,x_3$. From (\ref{coord.}) we have (for the curl in toroidal coordinates see for instance \cite {Bat})
\begin{equation}
\label{mag.f.}
B(x) = \frac{a'(r)}{R+r \cos \alpha} \hat{e}_\alpha.
\end{equation}
We make for convenience the   explicit choice
\begin{equation}
\label{a(r)}
a(r) = (r-r_0)^{-\tau}, \qquad  \tau >\frac72, \qquad \textnormal{for} \qquad r_0< r < r_0+ \frac{R-r_0}{8}.
\end{equation}

\medskip

We note that other  choices of the magnetic field are possible, as for example one directed along $\hat e_\theta$
and singular on the border of the torus.
We have considered the form (\ref{mag.f.})  which gives more difficulties in the analysis, while other choices can be studied
following the same lines of the present treatise. In any case, it will be clear in what follows that the magnetic field has to be chosen with a suitably strong divergence on the boundary of $\Gamma.$

 A remark, that will play an important role in the sequel, is the  fact the magnetic force $ V(t)\wedge B(X(t))$ does not change the modulus of the velocity, since
\begin{equation}
\begin{split}
\frac{d}{dt}V^2=2V\cdot \dot{V}= 2V\cdot (E+V\wedge B)=2V\cdot E.  \label{baa}
\end{split}
\end{equation} 

We denote by $\Lambda_i$ the spatial support of $f_{i,0}(x,v)$ for any $i=1,\dots, n.$ Moreover $\partial \Gamma$ represents the border of the torus and  $\Gamma^c=\mathbb{R}^3\setminus\Gamma.$  

\noindent One of the main results of the present paper is the following theorem, which is valid for charge densities
$\sigma_i$ of any sign (both positive and negative).

\begin{theorem}
Let us fix an arbitrary positive time T. Consider the initial data $f_{i,0} \in L^\infty$ such that $\Lambda_i\subset \Gamma^c\setminus \partial\Gamma$, with a distance between $\Lambda_i$ and $\Gamma$ greater than
$d_0>0$.
 Let  $f_{i,0}$ also  satisfy the following hypotheses:
\begin{equation}
0\leq f_{i,0}(x,v)\leq C_0 e^{- \lambda |v|^{q}}g(|x|), \qquad 
q>\frac{18}{7}
\label{Ga1}
 \end{equation}
 where  $g$ is a bounded, continuous, not increasing  function such that, for $|x|\geq 1$,
\begin{equation}
g(|x|) \leq \frac{C}{|x|^{\alpha}} \qquad \qquad \textit{with}\,\,\,\, \alpha>3
\label{asp1}
\end{equation}
  being $\lambda,$ $C_0$ and $C$  positive constants. Then $\forall (x,v)$   there exists a solution to equations (\ref{ch}) on $[0, T]$
  such that $X(t)\in \Gamma^c\setminus \partial\Gamma.$

Moreover there exist   positive constants $C'$ and ${\lambda}'$ such that  
 \begin{equation}
 0\leq f_i(x,v,t)\leq C'e^{- \lambda' |v|^{q}}g(|x|).
 \label{dec1}
 \end{equation}
This solution is unique in the class of those satisfying (\ref{dec1}).
\label{th_02}
\end{theorem}
The proof will be discussed in the next section.
We remark that the assumption on the  super-gaussian decay of the velocities (\ref{Ga1})  is related to the fact that the magnetic lines are not straight lines, which is a consequence of the boundedness of $\Gamma.$  When studying unbounded regions (see  Section \ref{General}) we are able to consider data having a gaussian decay in the velocities.

The spatial decay (\ref{asp1}) implies that the total masses are bounded:
\begin{equation}
\label{mass_i}
M_i = \int_{\R^3 \setminus \Gamma} \, dx \int_{\R^3}  \, dv \, f_{i,0}(x,v ) < \infty.
\end{equation}

\bigskip

We would like to study cases in which some species has infinite mass, in order  to put in evidence whether or not the results are dependent on the boundedness of $M_i.$ Unfortunately, we are able to do this investigation only if the charges of all the  species have the same sign (all positive or all negative). In the proof we will discuss why. 

The following theorem holds:

\begin{theorem}
Let us fix an arbitrary positive time T.   Assume that all $\sigma_i > 0$, or all $\sigma_i < 0$, $i=1,\dots, n$.
Consider the initial data $f_{i,0} \in L^\infty$ such that $\Lambda_i\subset \Gamma^c\setminus \partial\Gamma$, with a distance between $\Lambda_i$ and $\Gamma$ greater than
$d_0>0$.
 Let  $f_{i,0}$ also  satisfy the following hypotheses: 
\begin{equation}
0\leq f_{i,0}(x,v)\leq C_0 e^{- \lambda |v|^{q}}g(|x|), \qquad 
 q >  \frac{45}{7} - \frac97 \alpha  \quad \textit{with} \,\,\,\, \frac83<\alpha\leq 3
\label{Ga}
 \end{equation}
 where  $g$ is a bounded, continuous, not increasing  function such that, for $|x|\geq 1$,
\begin{equation}
g(|x|) \leq \frac{C}{|x|^{\alpha}}
\label{asp}
\end{equation}
  being $\lambda,$ $C_0$ and $C$  positive constants. Then $\forall (x,v)$ there exists a solution to equations (\ref{ch}) on $[0, T]$
  such that 
$X(t)\in \Gamma^c\setminus \partial\Gamma.$

Moreover there exist   positive constants $C'$ and ${\lambda}'$ such that  
 \begin{equation}
 0\leq f_i(x,v,t)\leq C'e^{- \lambda' |v|^{q}}g(|x|).
 \label{dec}
 \end{equation}
This solution is unique in the class of those satisfying (\ref{dec}).
\label{th_03}
\end{theorem}

The proof will be discussed in the next section.
We remark that the lower bound on $\alpha$ in (\ref{Ga}) is due to technical reasons, i.e. to control the tail of the velocities  distribution
in the iterative method of subsection \ref{subs_conv}, while the upper bound is due to allow for infinite mass. The border case
$\alpha=3$   should be considered apart, since some estimates change slightly (as (\ref{Q^})), but it is
simpler than $\alpha<3$ and will not be investigated explicitly.

\bigskip

\noindent The Vlasov-Poisson equation has been widely studied in several papers. The problem of existence and uniqueness of the solution in ${\mathbb{R}}^3$, for $L^1$ data, is completely solved since many years (see \cite{L, Pf, S,W} and a nice review of many such results in \cite{G}). 
For the one species case it has been studied also the case of spatial density not belonging to $L^1$ to outline that the results  do not depend on the finite total mass. In particular  it has been first investigated a system with simplified bounded interaction  \cite{CCMP} and more recently for coulomb interaction (or coulomb like interaction) in \cite{CCM2, Rem, R3,inf,CCM16}. Other papers related to these problems, with many different optics, are    \cite{Ca, Ch, ChZ,  DMS, GJP,  J, Lo,  MMP_11, Pa1,  Pa2, P,  Sal09,  S1, S2, S3} to quote some of them. 
In particular in some papers the confinement of the plasma by a singular magnetic field is discussed \cite {CCM12,CCM2,Rem, CCM15, inf}.  In \cite {CCM12} and  \cite {CCM15} we have discussed the confinement in
an infinite cylinder and in a torus of a plasma with bounded total mass and velocities, in \cite {CCM2,Rem} the confinement in an infinite cylinder of a plasma always with bounded velocities but now with infinite mass,  while  \cite {inf} deals with the confinement in an infinite cylinder of a plasma with infinite mass and unbounded velocities.

\noindent In the present paper we discuss the opposite problem, that is how to forbid the entrance of an external plasma
inside a fixed region by means of a singular magnetic field. Thus, we conjugate the problems of confinement and infinite mass in the whole of $\mathbb{R}^3$ while, in the preceding articles \cite {CCM2,Rem} and \cite {inf} ,  we dealt separately with confined plasmas in a cylinder (an almost one-dimensional system) or free plasmas in the space. For this reason we will borrow some of the techniques already used in both cases, adapting them to the actual problem and repeating some of them to make the present paper self-contained. On the other side,  we  discuss in more detail some aspects that we have not yet investigated.

\noindent The novelties with respect to the preceding literature consist in the following aspects. First, in Theorem \ref{th_02}
we study a system with different charge sign species. The analysis is done only in case of finite total mass, since we are not able to deal
with infinite masses, by the lack of some properties of the local energy (\ref{e}).
Secondly, in Theorem \ref{th_03} we deal with a one species infinite mass plasma with a curved obstacle, whose geometry leads to some difficulties, in particular in the proof of the estimate of the electric field, which is given in Proposition \ref{prop_1}.

\noindent The plan of the paper is the following: in Section 2 we prove Theorem \ref{th_02}
and Theorem \ref{th_03},   Section \ref{General} is devoted to some generalizations, and 
in the Appendix we put some technical proofs.

\section{Proofs}

\bigskip

 \subsection{Some comments on the Theorems}

In Theorem \ref{th_02} the $n$ species have charges of different sign, and this difficulty imposes to study
initial data with finite total mass, while in Theorem \ref{th_03} the $n$ species have charges of the same sign,
and this simplification permits to extend the hypothesis to considering initial data with infinite mass. Moreover
in this second case the fact that the species are $n$ or only one  does not differ substantially, and we will
investigate explicitly the case of a single species.
In both cases the strategy of the proof relies on the introduction of a truncated dynamics, denoted as \textit{partial dynamics}, in which we assume the initial data $f_{i,0}$ having compact support in the velocities, that is initially
$|v|\leq N$. We investigate this case and then we allow $N\to \infty$.

 \bigskip
 
 We premise a remark about the case considered in Theorem \ref{th_02}.
 We put
 $$
f(x,v,t)=\sum_{i=1}^n f_i(x,v,t)   \qquad \textnormal{and} \qquad
\rho(x,t) = \sum_{i=1}^n \sigma_i \, \rho_i(x,t).
$$
  The total energy of the plasma is
 \begin{equation}
 \frac12 \int d x \int d v \, |v|^2 f( x, v,t) +\frac12 \int dx \,
 \rho( x,t)\int dy \, 
\frac{ \rho ( y,t)}{  | x- y|}  = C, 
\label{finite_energy}
\end{equation}
where the first integral is the kinetic energy and the second one the potential energy.
The total energy is finite (because of the finite total mass) and it is a constant of the motion. 
We want to recall the well known fact that
$$
 \int d x \int dy \, \rho(x,t) \, 
\rho ( y,t)  \frac{1}{| x- y|}  = \int dx \, |E(x, t)|^2,
$$
and hence the potential energy is positive, in spite of the fact that the spatial density $\rho$ is not definite in sign.
We put
$$
\Phi(x,t) = \int dy \, 
\rho ( y,t)  \frac{1}{| x- y|}. 
$$
Since $\rho(x,t) = \textnormal{div} E(x,t)$, we have
$$
 \int d x \int dy \, \rho(x,t) \, 
\rho ( y,t)  \frac{1}{| x- y|} =
\int dx \, \Phi (x,t) \, \textnormal{div} E(x,t) =
$$
$$
\int dx \big[ \textnormal{div}(E(x,t) \Phi(x,t))  - E(x,t) \cdot \nabla \Phi(x,t)   \big] =
$$
$$
\int dx \, \textnormal{div}(E(x,t) \Phi(x,t)) + \int dx \, |E(x,t)|^2.
$$
The first integral vanishes, as it can be seen performing the integral over a ball of radius $R$, and taking
the limit $R\to \infty$. Indeed, for the Gauss theorem  it is
$$
\int dx \, \textnormal{div}(E(x,t) \Phi(x,t))=\int (E(x,t)  \Phi(x,t)) \cdot \hat n \, dS
$$
where the right-hand integral is taken over the surface of the ball. If $\rho\in L^1$, $|E|$ goes to zero as $\frac{1}{|x|^2}$
for large $|x|$, and $ \Phi$ as $\frac{1}{|x|}$.  The surface grows as $R^2$, hence the integral goes to zero.
Consequently eq. (\ref{finite_energy}) becomes 
\begin{equation}
\frac12 \int d x \int d v \, |v|^2 f( x, v,t) +\frac12 \int dx \,
 |E(x,t)|^2  = C,
\end{equation}
and hence both potential  and kinetic energy are bounded and positive during the motion.
This consideration allows to treat the case contemplated in Theorem \ref{th_02} (species with charges of different sign)
along the same lines of Theorem \ref{th_03}. This last theorem allows for infinite mass,
which is difficult in the case of Theorem \ref{th_02}, since we are not able to prove Proposition {\ref{propo}}
(strongly based on the positivity of the potential), which is the key tool to control infinite mass.        
Therefore in the sequel we study explicitly  Theorem \ref{th_03}, because it is the more
complicated case.  
 During the proof we will see that the proof of Theorem \ref{th_02}  will follow from that of
Theorem \ref{th_03}, since we can bound the absolute value of the electric field acting on a characteristic
by the sum of the absolute values of the electric fields produced by the
different species; moreover in several points (as (\ref{K}), (\ref{5/3}),  (\ref{i5}), and the omitted proof of 
Lemma \ref{Lem_1}) we need to bound the kinetic energy of the plasma, which in case of Theorem {\ref{th_02}} is simply bounded by a constant, for what seen before.

\subsection{Partial dynamics and local energy in case of Theorem \ref{th_03}}

For what said before we consider hereafter the case of species with charges of the same sign. Moreover the number
of species will be assumed to be one, since all the main bounds proved in the sequel, as for example the fundamental bound of the electric field (\ref{E_field}),  are not affected by the number
of species (growing linearly in the number of species $n$, which is finite). Hence from now on we will drop the index of species $i=1,\dots, n$.

Beside system (\ref{ch}) we introduce a modified differential system, called \textit{partial dynamics}, in which the initial density has compact support in the velocities. More precisely, for any positive integer $N,$  we  consider the following equations:
\begin{equation}
\label{chN}
\begin{cases}
\dsp  \dot{X}^N(t)= V^N(t)\\
\dsp  \dot{V}^N(t)=  E^N(X^N(t),t) + V^N(t) \wedge B^N(X^N(t))  \\
\dsp ( X^N(t'), V^N(t'))=(x,v)  \\
\dsp  f^N(X^N(t),V^N(t),t)=f^N_0(x,v),
 \end{cases}
\end{equation}
where
\begin{equation*}
\left(X^N(t),V^N(t)\right)=   
\left((X^N(t,t',x,v),V^N(t,t',x,v)\right)
\end{equation*} 
\begin{equation*}
E^N(x,t)=\int \frac{x-y}{|x-y|^3} \ \rho^N(y,t) \, dy
\end{equation*} 
\begin{equation*}
\rho^N(x,t)=\int f^N(x,v,t) \, dv,
\end{equation*} 
being the initial distribution $f_0^N$ defined as
\begin{equation}
f^N_0(x,v)=f_0(x,v)\chi\left(|v|\leq N\right) \label{ic}
\end{equation}
with  $\chi\left(\cdot\right)$  the characteristic function of the set $\left(\cdot\right)$. 

\bigskip

 For a single species, the existence of the solution to (\ref{chN}) 
in case the spatial density satisfies $\rho(x, 0)\leq C |x|^{-\alpha},$ with $\alpha$ in a suitable
range to include infinite mass, follows from what already done in \cite{CCM2}
and \cite{Rem}.

In what follows some positive constants will appear, generally denoted by $C$
(changing from line to line).  All of them will depend exclusively on  the initial data and an arbitrarily fixed, once for ever, time $T$,  but not on $N$,
 while any dependence on $N$ in the estimates will be clearly stressed. 

\bigskip

We introduce the {\textit{local energy}}, which is  a fundamental tool to deal with the infinite charge of the plasma. 
 For any vector $\mu\in \mathbb{R}^3$ and any $R>1$ we define the function:
\begin{equation}
\varphi^{\mu,R}( x)=\varphi\Bigg(\frac{| x-\mu|}{R}\Bigg) \label{a}
\end{equation}
with $\varphi$ a smooth function such that:
\begin{equation}
\varphi(r)=1 \ \ \hbox{if} \ \ r\in[0,1] \label{b}
\end{equation}
\begin{equation}
\varphi(r)=0 \ \ \hbox{if} \ \ r\in[2,+\infty) \label{c}
\end{equation}
\begin{equation}
-2\leq \varphi'(r)\leq 0.\label{d}
\end{equation}
We define the local energy as:
\begin{equation}
\begin{split}
& W^N( \mu,R,t)=\\&\frac12 \int_{\R^3\setminus \Gamma} d x\, \varphi^{\mu,R}( x)\left[\int d v \ |v|^2 f^N( x, v,t)\,+
 \rho^N( x,t)\int_{\R^3\setminus \Gamma} dy \ 
\frac{ \rho^N ( y,t)}{  | x- y|}\right].
\end{split} \label{e}
\end{equation}
 It depends on the property of $f^N$ whether or not $W^N$ is bounded. For the moment, we stress that the local energy takes into account the complete interaction with the rest of the plasma out of the sphere of center $\mu$ and radius $2R.$

We set

\begin{equation}
Q^N(R,t)=\sup_{\mu\in {\mathbb {R}}^3 }W^N( \mu,R,t).\label{Q}
\end{equation}

By the definition of local energy, it is easy to prove, see for instance \cite{CCM16}, that
 
 \begin{lemma}
In the  hypotheses of Theorem \ref{th_03}, $\forall  R>0$ it holds
\begin{equation}
Q^N(R,0)\leq CR^{3-\alpha}.  \label{Q^}
\end{equation}
\label{Qbound}
\end{lemma}

 Let us introduce the following functions, for any $t\in [0,T]:$ 
\begin{equation}
{\mathcal{V}}^N(t)=\max
\left\{ \widetilde C,  \sup_{s\in [0,t]}\sup_{(x,v)\in \mathbb{R}^3 \times b(N)}|V^N(x,v,s)| \right\}, \label{mv}
\end{equation}
\begin{equation}
R^N(t)= 1 +\int_0^t {\mathcal{V}}^N(s)\, ds,
\label{RN}
\end{equation}
where $\widetilde{C}>1$ is a positive constant chosen suitably large for further technical
reasons, 
and  $b(N)$ is the ball in $\mathbb{R}^3$ of center $0$ and radius $N.$

\medskip

 The main result we can prove on the local energy is stated in the following Proposition, whose proof is contained in \cite{CCM16} and not repeated here.

\begin{proposition}
In the  hypotheses of Theorem \ref{th_03}, for any $ t\in [0,T]$ it holds
\begin{equation*}
Q^N(R^N(t), t)\leq CQ^N(R^N(t),0).
\end{equation*}
\label{propo}
\end{proposition}

\subsection{Preliminary results}\label{subs_res}

As said before, we refer here to a single species, since the same reasoning in the following can be repeated
for each of the $n$ species.

\bigskip

We want to remove the cut-off $N$ on the velocities, by showing that the velocity of any characteristic can be bounded over the time interval $[0,T],$ independently of $N.$ A naive approach would be to look for a bound on the electric field valid for any time. Unfortunately we are not able to do this, but only to bound the time average of the electric field acting on a characteristic, as we will show in Propositions \ref{prop_1} and \ref{media},  and this will be enough to achieve the result.
Such estimates are improvements with respect to the analogous ones obtained in \cite{CCM15}.

\noindent A fundamental tool in the proof of Theorem \ref{th_03} is the achievement of the following estimate on the time integral
of the electric field,  along the same lines of \cite{CCM16}.
\begin{proposition}
In the hypothesis of Theorem \ref{th_02} (for $\alpha>3$) and Theorem \ref{th_03} (for $\alpha\leq 3$) respectively, we have
\begin{equation}
\int_0^t |E^N(X^N(s), s)|\, ds \leq C \mathcal{A}(\mathcal{V}^N(t)),
\label{E_field}
\end{equation}
where
\begin{equation}
\mathcal{A}(\mathcal{V}^N(t))=
\left\{
\begin{aligned}
&C \left[\mathcal{V}^N(t) \right]^{\frac67} \log^2 \mathcal{V}^N(t) \qquad \textit{for} \,\,\,\,\alpha>3 \\
&C \left[\mathcal{V}^N(t)\right]^{\frac67} \log^3 \mathcal{V}^N(t)   \qquad \textit{for} \,\,\,\,\alpha=3  \\
&C \left[\mathcal{V}^N(t)\right]^{\gamma} \log^2 \mathcal{V}^N(t)   \qquad \textit{for} \,\,\,\,\frac{8}{3}<\alpha<3 
\end{aligned}
\right.
\end{equation}
and $\gamma = \frac{15}{7} - \frac{3}{7}\alpha$.
\label{prop_1}
\end{proposition}
\noindent We give the proof in the more complicated case $\frac{8}{3}<\alpha<3$, the other cases
being simpler.
The exponent $\gamma$ above is necessary in order to deal
with the exponential tail (\ref{dec})  of the velocities distribution.  

The proof of Proposition  \ref{prop_1} is a direct consequence of the following:

\begin{proposition}
Setting
\begin{equation*}
\langle E^N \rangle_{{\Delta}}
:= \frac{1}{{\Delta}} \int_{t}^{t+{\Delta}} |E^N(X(s), s)| \, ds,
\end{equation*}
there exists a positive number ${\Delta},$ depending on $N,$ such that:
\begin{equation}
\langle E^N \rangle_{{\Delta}} \leq
C {\cal{V}}^N(T)^{\gamma} \log^2({\mathcal{V}}^N(T)),
\label{averna}
\end{equation}
for any $t\in [0,T]$ such that $t\leq T- \Delta$.
\label{media}
\end{proposition}

\medskip

We observe that
Proposition \ref{media} is sufficient to achieve the proof of Proposition \ref{prop_1},  which can be done by dividing the interval $[0,T]$ in $n$ subintervals $[t_{i-1},t_{i}]$,  $i=1,...,n$, with
$t_0=0$,  $t_n=T,$ such that ${\Delta}/2\leq t_{i-1}-t_{i}\leq {\Delta},$ and using Proposition \ref{media} on each of them.
We postpone the proof of Proposition \ref{media} to the Appendix.

As a consequence of Proposition \ref{prop_1},
the following holds:
\begin{corollary}
\begin{equation}
{\mathcal{V}}^N(T) \leq C N
\label{V^N}
\end{equation}
\begin{equation}
\rho^N(x, t) \leq C N^{3\gamma} \log^6 N
\label{rho_t}
\end{equation}
\begin{equation}
Q^N(R^N(t), t)\leq CN^{3-\alpha}
\label{Q^N}
\end{equation}
\begin{equation}
|B(X^N(t))|\leq CN^{\frac{2(\tau+1)}{\tau}},
\label{B}
\end{equation}
being $\gamma$ the exponent in Proposition \ref{prop_1}  and $\tau$ the one in eq. (\ref{a(r)}).
\label{coro}\end{corollary}

\begin{proof}
To prove (\ref{V^N}) we observe that the external Lorentz force does not affect the modulus of the particle velocities, being
\begin{equation}
\frac{d}{dt}\frac{|V^N(t)|^2}{2}=V^N(t) \cdot  E^N(X^N(t), t). \label{magn}
\end{equation}

Hence
\begin{equation}
|V^N(t)|^2=|v|^2+2\int_0^tV^N(s) \cdot  E^N(X^N(s), s)\,ds.\label{v^2}
\end{equation}

This fact, by Proposition \ref{prop_1} and the choice of the initial data such that $|v|\leq N,$ implies 
\begin{equation}
|V^N(t)|^2 \leq N^2+C \left[{\mathcal{V}}^N(t)\right]^{\gamma+1} \log^2({\mathcal{V}}^N(t)).
\end{equation} 
Hence, since $\gamma+1<2,$ by taking the $\sup_{t\in [0,T]}$ we obtain the thesis.

 Now we prove (\ref{rho_t}). Putting
 $$
 (\bar{x}, \bar{v})=\left(X^N(x,v,t),V^N(x,v,t)\right),$$
by using the invariance of the density along the characteristics we have 
\begin{equation*}
\rho^N(\bar{x}, t) =\int f(\bar{x}, \bar{v},t)\, d\bar{v}=\int f_0\left(x(\bar x, \bar v, t),v(\bar x, \bar v, t)\right)\, d\bar{v}.
\end{equation*}
We notice that, putting
$$
\widetilde{V}^N(t)=\sup_{0\leq s\leq t}|V^N(s)|,
$$
 from (\ref{v^2}), Proposition  \ref{prop_1} and (\ref{V^N}) it follows, for some positive constant $\bar{C}$
\begin{equation}
\begin{split}
|v|^2\geq |V^N(t)|^2- \bar{C} \widetilde{V}^N(t) N^{\gamma}\log^2 N.
\end{split}\label{VN}
\end{equation}
Hence, we decompose the integral as follows (writing for conciseness $x, v$ in place of $x(\bar x, \bar v, t),
v(\bar x, \bar v, t)$),
\begin{equation}
\begin{split}
&\rho^N(\bar{x}, t)\leq\\&\int_{|\bar v|\leq 2\bar{C} N^{\gamma}\log^2 N} f_0(x,v)\ d\bar{v}+C \int_{|\bar v|> 2\bar{C} N^{\gamma}\log^2 N} e^{-\lambda |v|^2}\ d\bar{v}\\& \leq   CN^{3\gamma}\log^6 N+  C \int_{|\bar v|> 2\bar{C} N^{\gamma}\log^2 N} e^{-\lambda |v|^2}\ d\bar{v}.   \label{f_0}
\end{split}
\end{equation}
Being by definition $|\bar{v}|\leq \widetilde{V}^N(t)$,
if $\widetilde{V}^N(t)\geq |\bar v| > 2\bar{C} N^{\gamma}\log^2 N$,  then by (\ref{VN}),
  $$|v|^2\geq |V^N(t)|^2-\frac{[ \widetilde{V}^N(t)]^2}{2}.$$
Since the inequality holds for any $t\in [0,T],$ it holds also at the time in which $V^N$ reaches its maximal value over $[0,t],$ that is 
$$
|v|^2 \geq  [\widetilde{V}^N(t)]^2-\frac{[ \widetilde{V}^N(t)]^2}{2}\geq \frac{| {V}^N(t)|^2}{2}=\frac{|\bar v|^2}{2}.$$
Hence from (\ref{f_0}) it follows
\begin{equation}
\begin{split}
&\rho^N(\bar{x}, t)\leq CN^{3\gamma}\log^6 N +C \int e^{-\frac{\lambda}{2}|\bar{v}|^2}\ d\bar{v} \leq CN^{3\gamma}\log^6 N.\end{split}
\end{equation}

To obtain (\ref{Q^N}) we simply combine (\ref{V^N}) with Lemma \ref{Qbound}  and Proposition \ref{propo}.

We prove (\ref{B}) by an  argument analogous to the one used in \cite{CCM15}.

\noindent We consider the time evolution of a single characteristic of the plasma in toroidal
coordinates (see (\ref{tor_coor})), in particular the second equation of that system,  
\begin{equation}
(R+r \cos\alpha) \ddot \theta + 2 \, \dot r \, \dot \theta   \cos\alpha 
-2 \, \dot\alpha \, \dot\theta \, r \sin\alpha =
E_\theta  - \frac{a'(r) \, \dot r}{R+r \cos\alpha},
\end{equation}
which, after multiplying by $(R+r \cos\alpha)$, becomes (following the notation in (\ref{tor_vel})),

\begin{equation}
a'(r) \dot r = -(R+r \cos\alpha)^2 \ddot \theta - 2 v_r v_\theta \cos\alpha
+ 2 v_\alpha v_\theta \sin\alpha + (R+r \cos\alpha) E_\theta.
\label{confin}
\end{equation}
Integrating in time (\ref{confin}), for the left hand side we have
\begin{equation}
\int_0^t a'(r(s)) \frac{dr}{ds} \, ds = a(r(t)) - a(r(0)),
\label{lhs}
\end{equation}
while for the right hand side we get
\begin{eqnarray}
&&\int_0^t \left( -(R+r \cos\alpha)^2 \ddot \theta - 2 v_r v_\theta \cos\alpha
+ 2 v_\alpha v_\theta \sin\alpha + (R+r \cos\alpha) E_\theta \right) ds
 \nonumber \\
&&= -\left[ (R+r \cos\alpha)^2 \dot \theta \right]_0^t 
+\int_0^t \left( 2 \dot\theta (R+r \cos\alpha) (\dot r \cos\alpha - r \dot\alpha \sin\alpha)    \right) ds \,   \nonumber  \\
&&\,\,\,\,\,\,+ \int_0^t \left( -2 v_r v_\theta \cos\alpha + 2 v_\alpha v_\theta \sin\alpha
+ (R+r \cos\alpha) E_\theta \right) ds.
\label{rhs}
\end{eqnarray}
As a by-product we have the impenetrability of the region $\Gamma$ for the partial dynamics, since (\ref{lhs}) diverges to $\infty$ for $r\to r_0^+$ (border of the torus), while (\ref{rhs})  stays finite.      
We have indeed, by the bound on the maximum velocity (\ref{V^N}), and by the a priori estimate on the electric field
(\ref{C1}) (together with (\ref{Q^N})),
$$
a(r(t)) \leq a(r(0)) + C N + C N^2.
$$ 
Since $a(r(0))$ depends only on the initial conditions, and it is finite by the assumption given in Theorem \ref{th_02}
and \ref{th_03}  (that the distance between $\Lambda_i$ and $\Gamma$ is greater than
$d_0>0$),    then by (\ref{mag.f.}) and the choice of the function $a(r)$ (\ref{a(r)}) we obtain (\ref{B}).

\end{proof}

\bigskip

Before going into the proof of Theorem \ref{th_03} we give some technical estimates, whose proofs can be found in \cite{CCM16}. 
\begin{proposition}
\begin{equation}
\|E^N(t)\|_{L^\infty}\leq C {\mathcal{V}}^N(t)^{\frac43}Q^N(R^N(t),t)^{\frac13}.
\label{C1}
\end{equation}
\label{prop2}\end{proposition}

\begin{lemma}
\label{lemR'/R}
For any $\mu\in \mathbb{R}$ and for any couple of positive numbers $R<R'$ we have:
$$
W^N(\mu,R',t)\leq C \left(\frac{R'}{R}\right)^3 Q^N(R,t).
$$
\end{lemma}

\begin{lemma}
For any $\mu\in \mathbb{R}^3$ and $R>0$ it holds
\begin{equation*}
\int_{1\leq |\mu-x|\leq R}\frac{\rho^N(x,t)}{|\mu-x|^2}\,dx\leq CQ^N(R,t)^{\frac35}.
\end{equation*}
\label{Lem_1}
\end{lemma}

\begin{lemma}
Let $R^N(t)$ be defined in (\ref{RN}). Then for any $\mu\in \mathbb{R}^3$ it holds
  \begin{equation}
\int_{3R^N(t)\leq |\mu-x| }\frac{\rho^N(x,t)}{|\mu-x|^2}\,dx\leq C.
\end{equation}\label{Lem_1'}
\end{lemma}

 \subsection{Convergence of the partial dynamics as $N \to \infty$   and proof of Theorem \ref{th_03}}
 \label{subs_conv}

We now start the proof of Theorem \ref{th_03}. We recall that $\Lambda$ is the spatial support of $f_0$ and $b(N)$ is the ball in $\mathbb{R}^3$ of center $0$ and radius $N.$ We fix a couple $(x,v)\in  \Lambda\times b(N)$,
 and we consider $X^N(t)$ and $X^{N+1}(t),$ that is the time evolved characteristics, both starting from this initial condition, in the  different dynamics relative to the initial distributions $f_0^N$ and $f_0^{N+1}.$
 
  We set
\begin{equation}
\delta^N(t) =\sup_{(x,v)\in \Lambda\times b(N)} |X^N(t)-X^{N+1}(t)|
\label{delta^N}
\end{equation}
\begin{equation}
\eta^N(t)=\sup_{(x,v)\in \Lambda\times b(N)} |V^N(t)-V^{N+1}(t)|
\end{equation}
and 
\begin{equation}
\sigma^N(t)=\delta^N(t)+\eta^N(t).
\end{equation}
Our goal is to prove the following estimate
\begin{equation}
\sigma^N(t)\leq C^{-N^c} \label{fff} 
\end{equation}
with $c$ a positive exponent.  Once we get this estimate, the proof of the Theorem is accomplished. Indeed, (\ref{fff}) implies that the sequences $X^N(t)$ and $V^N(t)$ are Cauchy sequences, uniformly on  $[0,T].$ 
 Therefore, for any fixed $(x,v)$ they converge to limit functions which we call $X(t)$ and $ V(t).$ 
 The proof of property (\ref{dec}),  and the facts that the solution is unique and satisfies
 system (\ref{ch}),  follow directly from what already done in Section 3 of \cite{CCM16}.
\bigskip

$\mathbf{Proof\, of \,estimate\, (\ref{fff})}$.

  We have
\begin{equation}
\begin{split}
&|X^N(t)-X^{N+1}(t)| =\\&  \, \bigg| \int_0^t dt_1 \int_0^{t_1} dt_2 \, \Big[E^N\left(X^N(t_2), t_2\right)
+V^N(t)\wedge B\left(X^{N}(t_2)\right)   \\
&  -  E^{N+1}\left(X^{N+1}(t_2), t_2\right)
-V^{N+1}(t)\wedge B\left(X^{N+1}(t_2)\right)  \Big] \bigg| \leq   \\
&\int_0^t dt_1 \int_0^{t_1} dt_2 \, \left[ \mathcal{G}_1(x,v,t_2) + \mathcal{G}_2(x,v,t_2) + \mathcal{G}_3(x,v,t_2)\right]
\end{split}
\label{iter_}
\end{equation}
where
\begin{equation}
\mathcal{G}_1(x,v,t) = \left|E^N\left(X^N(t), t\right)
-E^N\left(X^{N+1}(t), t\right)\right|
\end{equation}
\begin{equation}
\mathcal{G}_2(x,v,t) = \left|E^N\left(X^{N+1}(t), t\right)
-E^{N+1}\left(X^{N+1}(t), t\right)\right|
\end{equation}
and
\begin{equation}
\mathcal{G}_3(x,v,t) = \left|V^N(t)\wedge B\left(X^N(t)\right)
-V^{N+1}(t)\wedge B\left(X^{N+1}(t)\right)\right|. 
\end{equation}

To estimate the first term $\mathcal{G}_1$
we have to prove a quasi-Lipschitz property for $E^N.$ We consider the difference $|E^N(x,t)-E^N(y,t)|$ at two generic points $x$ and $y,$ (which will be in our particular case $X^N(t)$ and $X^{N+1}(t)$) and set: 
\medskip

\begin{equation*}
\mathcal{G}_1'(x,y,t)= |E^N(x,t)-E^N(y,t)|\chi (|x-y|\geq 1)
\end{equation*}
and
\begin{equation*}
\mathcal{G}_1''(x,y,t)= |E^N(x,t)-E^N(y,t)|\chi (|x-y|< 1).
\end{equation*}
 Hence 
\begin{equation}
\mathcal{G}_1(x,v,t)=\mathcal{G}_1'(X^N(t),X^{N+1}(t),t)+\mathcal{G}_1''(X^N(t),X^{N+1}(t),t). \label{F}\end{equation}
\medskip
By Proposition \ref{prop2} and Corollary \ref{coro} we have
\begin{equation}
\begin{split}
\mathcal{G}_1'(x,y,t)\leq&\ |E^N(x,t)|+|E^N(y,t)|\leq C N^{\frac{7-\alpha}{3}}.
\end{split}\label{b0}
\end{equation}
By the range of the parameter $\alpha\in(\frac83, 3)$  it is  $\frac{7-\alpha}{3}<3\gamma$, so that we get
\begin{equation}
\begin{split}
\mathcal{G}_1'(x,y,t)\leq CN^{3\gamma}\leq CN^{3\gamma}|x-y|. \label{b1}
\end{split}
\end{equation}

\medskip
Let us now estimate the term $\mathcal{G}_1''$. 

We put $\zeta=\frac{x+y}{2}$ and consider the sets: 

\begin{equation}
\begin{split}
&S_1=\{z:|\zeta-z|\leq 2|x-y|\}, \\& S_2=\Big\{z:2|x-y|\leq |\zeta-z|\leq\frac{4}{|x-y|}\Big\}\\& S_3=\Big\{z:|\zeta-z|\geq \frac{4}{|x-y|}\Big\}.\label{S'}
\end{split}
\end{equation}
We have:
\begin{equation}
\begin{split}
&\mathcal{G}_1''(x,y,t)\leq  \int_{S_1\cup S_2\cup S_3} \left|\frac{x-z}{|x-z|^3}-\frac{y-z}{|y-z|^3}\right| \rho^N(z,t) \,dz.
\end{split} \label{s}
\end{equation}

By (\ref{rho_t}) and the definition of $\zeta$ we have
\begin{equation}
\begin{split}
&\int_{S_1} \left| \frac{x-z}{|x-z|^3}-\frac{y-z}{|y-z|^3}\right| \rho^N(z,t) \,dz\leq\\& CN^{3\gamma} \log^6 N  \int_{S_1} \left|\frac{1}{|x-z|^2}+\frac{1}{|y-z|^2}  \right|\,dz\leq C N^{3\gamma} \log^6 N |x-y|.
\end{split}
\label{Ai1}\end{equation}
Let us pass to the integral over the set $S_2.$ By the Lagrange theorem
applied to each component   $i=1, 2, 3$,  of $E^N(x,t)-E^N(y,t)$, there exists $\xi_z$  such that $\xi_z=\kappa x +(1-\kappa)y$ and $\kappa\in [0,1]$ (depending on $z$), for which
\begin{equation} 
\begin{split}
\int_{S_2} \left|\frac{x_i-z_i}{|x-z|^3}-\frac{y_i-z_i}{|y-z|^3}\right|& \rho^N(z,t) dz\leq C |x-y|\int_{S_2} \frac{\rho^N(z,t)}{|\xi_z-z|^3}\,dz\\&\leq  CN^{3\gamma}\log^6 N |x-y|\int_{S_2} \frac{1}{|\xi_z-z|^3}\,dz.
\end{split}\label{D'}
\end{equation}
Since in $S_2$ it results $|\xi_z-z|\geq \frac{|\zeta-z|}{2}$, we have
$$
\int_{S_2} \frac{1}{|\xi_z-z|^3}\,dz \leq 8\int_{S_2} \frac{1}{|\zeta-z|^3}\,dz \leq C |\log |x-y|+1|
$$

\noindent and combining the results for the three components
\begin{equation}
\begin{split}
&\int_{S_2} \left|\frac{x-z}{|x-z|^3}-\frac{y-z}{|y-z|^3}\right| \rho^N(z,t) \,dz\leq  CN^{3\gamma} \log^6 N
|x-y|\, |\log |x-y|+1|. \label{D}
\end{split}
\end{equation}

For the last integral over $S_3,$ again by the Lagrange theorem, we have for some $\xi_z=\kappa x +(1-\kappa)y$ and $\kappa\in [0,1],$  
\begin{equation}
 \int_{ S_3} \left|\frac{x_i-z_i}{|x-z|^3}-\frac{y_i-z_i}{|y-z|^3}\right| \rho^N(z,t)\, dz\leq C |x-y| \int_{ S_3} \frac{\rho^N(z,t)}{|\xi_z-z|^3}\,dz.\label{S}
 \end{equation}
Notice that, since $z\in S_3$ and $|x-y|<1$ by definition of $\mathcal{G}_1''$, it is 
$$
|\xi_z - z|\geq  |\zeta - z| - |x-y|,
$$
and 
$$
\frac{1}{|\xi_z - z|} \leq \frac{1}{|\zeta - z| - |x-y|} \leq \frac{1}{\frac34|\zeta - z| +\frac14 |\zeta-z|- |x-y|}
\leq \frac43 \frac{1}{|\zeta-z|},
$$
 then we have
\begin{equation}
\begin{split}
\int_{ S_3} \frac{\rho^N(z,t)}{|\xi_z-z|^3}\,dz \leq C\int_{|\zeta-z|\geq 4 } \frac{\rho^N(z,t)}{|\zeta-z|^3}\,dz.
\end{split}
\end{equation}
Lemmas \ref{lemR'/R}, \ref{Lem_1} and \ref{Lem_1'}   imply
\begin{equation}
\int_{|\zeta-z|\geq 4 } \frac{\rho^N(z,t)}{|\zeta-z|^3}\,dz\leq CQ^N(R^N(t),t)^{\frac35} + C
\end{equation}
and by (\ref{Q^N}) we get
\begin{equation}
\int_{|\zeta-z|\geq 4} \frac{\rho^N(z,t)}{|\zeta-z|^3}\,dz\leq CN^{\frac{9-3\alpha}{5}}.
\end{equation}
Using this estimate in (\ref{S}) and going back to (\ref{s}), by (\ref{Ai1}) and (\ref{D}) we have proven that
\begin{equation}
\mathcal{G}_1''(x,y,t)\leq CN^{3\gamma} \log^6 N |x-y|  \, |\log |x-y|+1|. \label{f1}
\end{equation}

\bigskip
 In conclusion, recalling (\ref{F}), estimates (\ref{b1}),  (\ref{f1}), and the definition of $\delta^N(t)$ (\ref{delta^N}) show that
\begin{equation}
\mathcal{G}_1(x,v,t)\leq  CN^{3\gamma} \log^6 N \left(1+|\log \delta^N(t)|\right)\delta^N(t). \label{f1'}
\end{equation}

\medskip

We now concentrate on the term $\mathcal{G}_2$. 
Putting $\bar{X}=X^{N+1}(t),$  we have:
\begin{equation}
\mathcal{G}_2(x,v,t)\leq \mathcal{G}_2'(\bar X,t)+\mathcal{G}_2''(\bar X,t),
\label{I_2}
\end{equation}
where
\begin{equation}
\mathcal{G}_2'(\bar X,t)=\left|\int_{|\bar X - y|\leq 2\delta^N(t)} \frac{  \rho^N(y,t)-\rho^{N+1}(y,t)}{|\bar X-y|^2} \, dy \right|
\end{equation}
\begin{equation}
\mathcal{G}_2''(\bar X,t)=\left|\int_{2\delta^N(t)\leq |\bar X - y| } \frac{  \rho^N(y,t)-\rho^{N+1}(y,t)}{|\bar X-y|^2} \, dy \right|.
\end{equation}
By estimate (\ref{rho_t}) we get
\begin{equation}
\mathcal{G}_2'(\bar X,t)\leq \int_{|\bar X - y|\leq 2\delta^N(t)} \frac{  \rho^N(y,t)+\rho^{N+1}(y,t)}{|\bar X-y|^2} \, dy \leq CN^{3\gamma}\log^6 N \, \delta^N(t).\label{F2'}
\end{equation}
Now we estimate the term $\mathcal{G}_2''.$ By the invariance of $f^N(t)$ along the characteristics, making a change of variables, we decompose the integral as follows:
\begin{equation}
\begin{split}
\mathcal{G}_2''(\bar X,t)=  \left| \int_{S^N(t)} d y \, dw\,\frac{f_0^N( y, w)}{|\bar X-  Y^N(t)|^2}\ -\int_{S^{N+1}(t)}dy \, dw\,\frac{f_0^{N+1}(y,w)}{|\bar X-Y^{N+1}(t)|^{2}} \right|
  \end{split}
\label{I_2''}
\end{equation}
where we put, for $i=N, N+1,$ $$(Y^i(t), W^i(t))= (X^i(y,w,t), V^i(y,w,t))$$ 
$$
S^{i}(t)=\{( y, w):2\delta^N(t)\leq|{\bar X}- Y^i(t)| \}.
$$
We notice that it is
\begin{equation} 
\mathcal{G}_2''(\bar X,t)\leq \mathcal{I}_1+\mathcal{I}_2+\mathcal{I}_3\label{l_i}
\end{equation}
where
\begin{equation}
\mathcal{I}_1= \int_{S^N(t)} d y  
 \int d w \left| \frac{1}{|\bar X -  Y^N(t)|^{2}}
-\frac{1}{|\bar X -  Y^{N+1}(t)|^{2}} \right| f_0^N( y,w),
\end{equation}
\begin{equation}
\mathcal{I}_2= \int_{S^{N+1}(t)} d y   
  \int d w \, \frac{ \left| f_0^N( y,w) - f_0^{N+1}( y, w) \right|}{|\bar X -  Y^{N+1}(t)|^{2}}  \, 
  \end{equation}
 \begin{equation}
  \mathcal{I}_3=
  \int_{S^N(t)\setminus S^{N+1}(t)}dy\int dw\,\frac{f_0^N(y,w)}{\left|\bar{X}-Y^{N+1}(t)\right|^2}.
  \end{equation}
We start by estimating $\mathcal{I}_1.$ 
  By the Lagrange theorem 
 \begin{equation}
\mathcal{I}_1\leq C |Y^{N}(t)- Y^{N+1}(t)| \int_{S^N(t)} dy\int dw \, \frac{ f_0^N(y,w)}{|\bar X - \xi^N(t) |^{3}} 
\end{equation}
where $\xi^N(t)= \kappa Y^{N}(t)+ (1-\kappa) Y^{N+1}(t)$   for a suitable $\kappa \in [0, 1]$.
By putting 
\begin{equation*}
\left(\bar{y}, \bar{w}\right)= \left((Y^N(t), W^N(t)\right)    
\end{equation*}
and
$$
{\bar{S}}^N(t)=\left\{ (\bar y, \bar w): (y, w)\in  S^N(t) \right\},
$$
 we get
\begin{equation}
\begin{split}
\mathcal{I}_1\leq & \, C  \delta^N(t)\int_{ {\bar{S}}^N(t)}d\bar{y}\int d\bar{w} \, \frac{ f^N(\bar{y},\bar{w}, t)}{|\bar X - \xi^N(t) |^{3}}.\end{split}\label{bar}
\end{equation}

If $(y,w)\in S^N(t)$ then\begin{equation*}
 |\bar X -  \xi^N(t)| > |\bar X - \bar{y}| - |\bar{y} -Y^{N+1}(t)|\geq|\bar X - \bar{y}|-\delta^N(t)\geq \frac{|\bar X - \bar{y}|}{2}\label{>>}
\end{equation*}
which, by (\ref{bar}), implies
\begin{equation}
\begin{split}
\mathcal{I}_1 \leq \ &C \delta^N(t)\int_{ \bar S^N(t)}d\bar{y}\int d\bar{w} \, \frac{ f^N(\bar{y},\bar{w}, t)}{|\bar X - \bar{y} |^{3}}\,=C \delta^N(t)\int_{\bar S^N(t)}d\bar{y}\, \frac{ \rho^N(\bar{y}, t)}{|\bar X - \bar{y} |^{3}}
\end{split}.
\end{equation}
Now we consider the sets $$
A_1= \left\{(\bar y,\bar w):2 \delta^N(t)\leq\left|\bar{X}-\bar{y}\right|\leq  \frac{4}{\delta^N(t)}\right\}
$$
$$
A_2= \left\{(\bar y,\bar w): 1\leq\left|\bar{X}-\bar{y}\right|\leq 3R^N(t)\right\} $$
$$
 A_3=  \left\{(\bar y, \bar w):3R^N(t)\leq  |\bar{X}-\bar{y}|\right\}.
$$
Then it is $\bar S^N(t)\subset \bigcup_{i=1,2,3} A_i$ 
and
\begin{equation}
\mathcal{I}_1 \leq C \delta^N(t)\sum_{i=1}^3\int_{  A_i}d\bar{y}\frac{ \rho^N(\bar{y}, t)}{|\bar X - \bar{y} |^{3}}.
\end{equation}
We estimate the integral over $A_1$ as we did in (\ref{D'}), the one over $A_2$ by means of Lemma \ref{Lem_1} and the last one by means of Lemma \ref{Lem_1'}, yielding
\begin{equation}
\begin{split}
\mathcal{I}_1 \leq C\delta^N(t)\left[N^{3\gamma}\log^6 N \, |\log \delta^N(t)|+Q^N\left(3R^N(t),t\right)^{\frac35}+1\right].
\end{split}\label{i1}
\end{equation}
Lemma \ref{lemR'/R} implies that
\begin{equation}
\mathcal{I}_1  \leq C \delta^N(t)\left[N^{3\gamma}\log^6 N \, |\log \delta^N(t)|+Q^N\left(R^N(t),t\right)^{\frac35}+1\right]
\end{equation}
and by (\ref{Q^N}) in Corollary \ref{coro} we get
\begin{equation}
\begin{split}
\mathcal{I}_1  \leq CN^{3\gamma} \log^6 N \, \delta^N(t)\left[|\log \delta^N(t)|+1\right]\end{split}.\label{i1'}
\end{equation}

We estimate now the term $\mathcal{I}_2.$ By the definition of $f_0^i,$ for $i=N, N+1,$ and by (\ref{Ga}) it follows 
\begin{equation}
\begin{split}
&\mathcal{I}_2 =
 \int_{S^{N+1}(t)} d y   
  \int_{} d w\,  \frac{  f_0^{N+1}( y, w) }{|\bar X -  Y^{N+1}(t)|^{2}} \chi\left(N\leq |w|\leq N+1\right) \leq\\& C\,e^{-\lambda N^q}\int_{S^{N+1}(t)} dy \int_{} dw\,\, \frac{g(y) }{|\bar X - Y^{N+1}(t)|^{2}} \chi\left(|w|\leq N+1\right).
\end{split} \label{AA} \end{equation}
We evaluate the integral over $S^{N+1}(t)$ by considering the sets  $$
B_1=   \left\{(y,w): |\bar{X}-Y^{N+1}(t)|\leq4R^{N+1}(t)\right\} $$

$$
 B_2=  \left\{(y,w):4R^{N+1}(t)\leq  |\bar{X}-Y^{N+1}(t)|\right\},
$$
so that 
\begin{equation}
S^{N+1}(t)\subset \bigcup_{i=1,2}B_i. \label{si}
\end{equation}
By putting 
\begin{equation*}
(\bar{y}, \bar{w})=\left((Y^{N+1}(t), W^{N+1}(t)\right),
\end{equation*}
 by (\ref{V^N}) it is $|\bar{w}|\leq CN,$ so that we have
\begin{equation}
\begin{split}
&\int_{B_{1}} dy \int_{} dw\, \frac{g(y) }{|\bar X - Y^{N+1}(t)|^{2}}\chi\left(|w|\leq N+1\right)\leq\\&C \int_{|\bar{w}|\leq C N} d\bar{w}\int_{ |\bar{X}-\bar{y}|\leq 4R^{N+1}(t)} d\bar{y}\, \frac{1 }{|\bar X - \bar{y}|^{2}}\leq CN^3R^{N+1}(t)\leq CN^4 .\label{g5}
\end{split}\end{equation}

For the integral over the set $B_2,$ we observe that, if $|\bar{X}-Y^{N+1}(t)|\geq 4R^{N+1}(t),$ then $|\bar{X}-y|\geq 3R^{N+1}(t)$ and  
\begin{equation*}
|\bar{X}-Y^{N+1}(t)|\geq |\bar{X}-y|-R^{N+1}(t)\geq \frac23|\bar{X}-y|.
\end{equation*} Hence
\begin{equation}
\begin{split}
&\int_{B_2 } dy \int_{} dw\, \frac{g(y) }{|\bar X - Y^{N+1}(t)|^{2}}\chi\left(|w|\leq N+1\right)\leq \\&\frac94\int_{ |\bar{X}-y|\geq 3R^{N+1}(t)} dy \int_{|w|\leq N+1} dw\, \frac{g(y) }{|\bar X -y|^{2}}\leq\\& CN^3\int_{|\bar{X}-y|\geq 3R^{N+1}(t)}  dy\, \frac{g(y) }{|\bar X -y|^{2}}.
\end{split}\end{equation}
The last integral can be shown to be bounded by a constant, by using Lemma \ref{Lem_1'}.
 Hence we can conclude that
\begin{equation}
\int_{B_2 } dy \int_{} dw\, \frac{g(y) }{|\bar X - Y^{N+1}(t)|^{2}}\chi\left(|w|\leq N+1\right)\leq CN^3.\label{g6}
\end{equation}

Going back to (\ref{AA}), by (\ref{si}), (\ref{g5}) and (\ref{g6}) we have
\begin{equation}
\mathcal{I}_2\leq Ce^{-\lambda N^q}\left(N^4+N^3\right)\leq  Ce^{-\frac{\lambda}{2} N^q}.
\label{I}
\end{equation}
Finally, we estimate the term $\mathcal{I}_3.$ If $(y,w)\in S^N(t),$ then
$$
\left|\bar{X}-Y^{N+1}(t)\right|\geq \left|\bar{X}-Y^{N}(t)\right| -\delta^N(t)\geq \delta^N(t),
$$
so that
 \begin{equation}
\begin{split}
 & \mathcal{I}_3 =
  \int_{S^N(t)\setminus S^{N+1}(t)}dy\int dw\,\frac{f_0^N(y,w)}{\left|\bar{X}-Y^{N+1}(t)\right|^2}\leq\\&
\frac{1}{[\delta^N(t)]^2}  \int_{\left(S^{N+1}(t)\right)^c}dy\int dw\,f_0^N(y,w).
 \end{split} \end{equation}
 If $(y,w)\in \left(S^{N+1}(t)\right)^c,$ then $$\left|\bar{X}-Y^{N}(t)\right|\leq \left|\bar{X}-Y^{N+1}(t)\right|+ \delta^N(t)\leq 3\delta^N(t).$$ 
 Hence, by putting 
 \begin{equation*}
 (\bar{y}, \bar{w}) = \left(Y^N(t), W^N(t)\right)
 \end{equation*}
  we get
 \begin{equation}
\begin{split}
  \mathcal{I}_3\leq &
\,\frac{1}{[\delta^N(t)]^2}  \int_{\left|\bar{X}-\bar{y}\right|\leq 3\delta^N(t)} d\bar{y}\int d\bar{w}f^N(\bar{y}, \bar{w}, t)\leq\\&\frac{1}{[\delta^N(t)]^2}  \int_{ \left|\bar{X}-\bar{y}\right|\leq 3\delta^N(t)}d\bar{y}\,\rho(\bar{y},t)\leq \\& C\frac{N^{3\gamma}\log^6 N}{[\delta^N(t)]^2}  \int_{ \left|\bar{X}-\bar{y}\right|\leq 3\delta^N(t)}d\bar{y}\,\leq CN^{3\gamma}\log^6 N\,
\delta^N(t).\end{split}\label{I3}
  \end{equation}
Going back to (\ref{l_i}), by (\ref{i1'}), (\ref{I}) and (\ref{I3}) we have
\begin{equation}
\mathcal{G}_2''\leq C\left[N^{3\gamma}\log^6 N \, \delta^N(t)\left(|\log \delta^N(t)|+1\right)+e^{-\frac{\lambda}{2} N^q}\right], \label{F2''}
\end{equation}
so that, by (\ref{I_2}), (\ref{F2'}) and (\ref{F2''}), we get
\begin{equation}
\mathcal{G}_2\leq C\left[N^{3\gamma}\log^6 N \, \delta^N(t)\left(|\log \delta^N(t)|+1\right)+e^{-\frac{\lambda}{2} N^q}\right].\label{f}
\end{equation}
Hence, going back to the estimate (\ref{f1'}) of the term $\mathcal{G}_1$ we have that
\begin{equation}
\mathcal{G}_1+\mathcal{G}_2\leq C\left[N^{3\gamma}\log^6 N \, \delta^N(t)\left(|\log \delta^N(t)|+1\right)+e^{-\frac{\lambda}{2} N^q}\right].\label{g'}
\end{equation}

We need now a property, which  is easily seen to hold by convexity, that is, for any $s\in (0,1)$ and $a\in (0,1)$ it holds:
$$
 s(|\log s|+1)\leq s|\log a|+a.
 $$
Hence, for $\delta^N(t)< 1,$ we have, for any $a<1,$
$$
\mathcal{G}_1+\mathcal{G}_2\leq C\left[N^{3\gamma} \log^6 N \, \left( \delta^N(t)|\log a|+a\right)+e^{-\frac{\lambda}{2} N^q}\right].
$$
Therefore, in case $\delta^N(t)< 1,$ we choose $a =e^{-\lambda N^q}$,  yielding
\begin{equation}
\mathcal{G}_1+\mathcal{G}_2\leq C\left[N^{3\gamma+q}\log^6 N \, \delta^N(t)+e^{-\frac{\lambda}{2} N^q}\right].\label{fdue}
\end{equation}
In case that $\delta^N(t)\geq1$ we come back to (\ref{g'})  (valid for any $\delta^N(t)$)
and we insert
the bound $\delta^N(t)\leq C N$, obtaining
$$
\mathcal{G}_1+\mathcal{G}_2\leq C\left[N^{3\gamma} \log^7 N \,  \delta^N(t)+e^{-\frac{\lambda}{2} N^q}\right],
$$
which can be included in (\ref{fdue}).

Finally we estimate the term $\mathcal{G}_3.$ We have
\begin{equation}
\begin{split}
\mathcal{G}_3(x,v,t) \leq &\ |V^N(t)|| B(X^{N}(t))-B(X^{N+1}(t)|+\\&
\ |V^N(t)-V^{N+1}(t)|| B(X^{N+1}(t))|. 
\end{split}\end{equation}
By applying the Lagrange theorem we have
\begin{equation*}
| B(X^{N}(t))-B(X^{N+1}(t)|\leq C\frac{|X^{N}(t)-X^{N+1}(t)|}{|r_0- r_\xi |^{\tau+2}}
\end{equation*}
where $r_\xi$ is the $r$-coordinate (in toroidal  coordinates) of a point   $\xi^N(t)$  of the segment joining $X^N(t)$ and $X^{N+1}(t).$ Due to the bound (\ref{B}), it has to be $|r_0-r_\xi|\geq \frac{C}{N^{2/\tau}}.$
 Hence
 \begin{equation}
 | B(X^{N}(t))-B(X^{N+1}(t)|\leq CN^{\frac{2(\tau +2)}{\tau}}\delta^N(t).
\end{equation}
This, together with the bounds (\ref{V^N}) and (\ref{B}), imply
\begin{equation}
\begin{split}
\mathcal{G}_3(x,v,t)& \leq \,C\left[N^{\frac{3\tau+4}{\tau }} \delta^N(t)+N^{\frac{2(\tau+1)}{\tau}}\eta^N(t)\right].\\&
 \label{F3}
\end{split}
\end{equation}
At this point, going back to (\ref{iter_}), taking the supremum over the set $\{(x,v)\in \Lambda\times b(N)\},$ by (\ref{fdue}) and (\ref{F3}) we arrive at:
\begin{equation}
\begin{split}
\delta^N(t) \leq \,&C\left(N^{3\gamma+q}\log^6 N +N^{\frac{3\tau+4}{\tau }}\right) \int_0^t dt_1 \int_0^{t_1} dt_2\,\delta^N(t_2)+\\ &C \int_0^t dt_1 \int_0^{t_1} dt_2\,e ^{-\frac{\lambda}{2} N^q}+
CN^{\frac{2(\tau+1)}{\tau}}\int_0^t dt_1 \int_0^{t_1} dt_2 \ \eta^N(t_2),\label{dd}
\end{split}
\end{equation}
where in (\ref{f1'}) we have taken into account  the bound (\ref{V^N}), which gives $|\delta^N(t)|\leq CN.$
On the other side, by using the same method to estimate the quantity $\eta^N(t),$ we get, analogously
\begin{equation}
\begin{split}
\eta^N(t) \leq \,&C\left(N^{3\gamma+q}\log^6 N +N^{\frac{3\tau+4}{\tau }}\right) \int_0^t dt_1 \,\delta^N(t_1)+\\ &C \int_0^t dt_1 \,e ^{-\frac{\lambda}{2} N^q}+
CN^{\frac{2(\tau+1)}{\tau}}\int_0^t dt_1  \, \eta^N(t_1).
\end{split}
\end{equation}
Since obviously $\delta^N(t_1)\leq \int_0^{t_1}dt_2\,\eta^N(t_2)$ we get from the last eqn.
\begin{equation}
\begin{split}
&\eta^N(t) \leq C\left(N^{3\gamma+q}\log^6 N+N^{\frac{3\tau+4}{\tau }}\right) \int_0^t dt_1 \int_0^{t_1} dt_2\, \eta^N(t_2)+C \int_0^t dt_1\,e ^{-\frac{\lambda}{2} N^q}\\&+CN^{\frac{2(\tau+1)}{\tau}}\int_0^t dt_1  \ \eta^N(t_1).\label{ee}
\end{split}
\end{equation}

Hence, summing up (\ref{dd}) and (\ref{ee}), we have:

\begin{equation}
\begin{split}
\sigma^N(t)\leq \,&C\left(N^{3\gamma+q}\log^6 N+N^{\frac{3\tau+4}{\tau }}\right) \int_0^t dt_1 \int_0^{t_1} dt_2\,\sigma^N(t_2)+\\&CN^{\frac{2(\tau+1)}{\tau}}\int_0^t dt_1  \ \sigma^N(t_1)+C \left(t+\frac{t^2}{2}\right)e ^{-\frac{\lambda}{2} N^q}.
\end{split}\label{d+e}
\end{equation}
We take
\begin{equation}
\nu > \frac{3\gamma+q}{2},  \label{nu}  
\end{equation}
moreover by Proposition \ref{prop_1} it is $\frac67<\gamma<1$, and by (\ref{Ga}) it is $q>\frac{18}{7}$,
 hence choosing $\tau > \frac72$
we get 
\begin{equation}
\begin{split}
\sigma^N(t)\leq \,C&\Bigg[N^{2\nu} \int_0^t dt_1 \int_0^{t_1} dt_2\,\sigma^N(t_2)+N^{\nu}\int_0^t dt_1  \ \sigma^N(t_1)+\\& \left(t+\frac{t^2}{2}\right)e ^{-\frac{\lambda}{2} N^q}\Bigg].
\end{split}\label{d'}
\end{equation}

We insert in the integrals the same inequality  for $\sigma^N(t_1)$ and $\sigma^N(t_2)$ and iterate in time, up to $k$ iterations. By direct inspection, using in the last step the estimate $\sup_{t\in [0,T]}\sigma^N(t)\leq CN,$ we arrive at
\begin{equation}
\begin{split}
\sigma^N(t)\leq &\,CNe ^{-\frac{\lambda}{2} N^q}\sum_{i=1}^{k-1}C^i\sum_{j=0}^{i}\binom{i}{j}\frac{N^{2\nu j}t^{2j}N^{\nu (i-j)}t^{i-j}}{(2j+i-j)!} +\\&C^kN \sum_{j=0}^{k}\binom{k}{j}\frac{N^{2\nu j}t^{2j}N^{\nu (k-j)}t^{k-j}}{(2j+k-j)!}.\end{split}
\label{double_sum}
\end{equation}
 In the double sum of (\ref{double_sum}) we have all combinations of $j$ double time integrations of
 $e^{-\frac{\lambda}{2}N^q}$, yielding the factor  $N^{2\nu j} t^{2 j}$ and the contribution $2j$ in the factorial at
 denominator, and of ($i-j$) single time integrations of $e^{-\frac{\lambda}{2}N^q}$, yielding the factor $N^{\nu (i-j)} t^{i-j}$
 and the contribution ($i-j$) in the factorial at denominator; 
 the last single sum in (\ref{double_sum}) has the same structure, coming from terms of the iteration which avoid integration of $e^{-\frac{\lambda}{2}N^q}$.  Note that  in absence of the middle term (single time integration) in
 the right hand side of (\ref{d'}) we would obtain
 $$
 \sigma^N(t) \leq  N C^k N^{2\nu k} \frac{t^{2k}}{(2k)!} + N e^{-\frac{\lambda}{2}N^q}
 \sum_{i=1}^{k-1} C^i N^{2\nu i} \frac{t^{2 i}}{(2 i)!}
 $$
 as in the case of  \cite{CCM16}.

By putting
\begin{equation*}
S_k''=\sum_{i=1}^{k-1}C^i\sum_{j=0}^{i}\binom{i}{j}\frac{N^{\nu (i+j)}t^{i+j}}{(i+j)!}
\end{equation*}
and 
\begin{equation*}
S_k'=C^k\sum_{j=0}^{k}\binom{k}{j}\frac{N^{\nu (j+k)}t^{j+k}}{(j+k)!}
\end{equation*}
we get
\begin{equation}
\sigma^N(t)\leq \,CNe ^{-\frac{\lambda}{2} N^q}S''_k+CNS_k'.\label{summ}
\end{equation}
 
 We start by estimating $S''_k.$ Recalling that $\binom{i}{j}<2^i$ we get
\begin{equation}
\begin{split}
S''_k\leq &\,\sum_{i=1}^{k-1}2^i C^i\sum_{j=0}^{i}\frac{N^{\nu (i+j)}t^{i+j}}{(i+j)!}.  
\end{split}
\end{equation}
The use of the Stirling formula $a^nn^n\leq n!\leq b^nn^n$ for some constants $a,b>0$ yields:
\begin{equation}
S''_k\leq \,\sum_{i=1}^{k-1}2^i\sum_{j=0}^{i}\frac{N^{\nu(i+j)}(Ct)^{i+j}}{(i+j)^{i+j}}\leq \sum_{i=1}^{k-1}2^i\frac{N^{\nu i}(Ct)^{i}}{i^{i}}
\sum_{j=0}^{i}\frac{N^{\nu j}(Ct)^{j}}{j^{j}},
\end{equation}
from which it follows, again by the Stirling formula,
\begin{equation}
S''_k\leq \left(e^{N^{\nu}Ct}\right)^2\leq e^{N^{\nu}C}.\label{sum'}
\end{equation}

For the term $S_k',$ putting $j+k=\ell,$ we get
\begin{equation}
S_k'\leq 2^k C^k\sum_{\ell=k}^{2k}\frac{N^{\ell\nu}(Ct)^{\ell}}{\ell^\ell}\leq  C^kk\frac{N^{2k\nu}}{k^k}.
\end{equation}
By choosing $k=N^\zeta$  with $\zeta>  2\nu$,  we have, for sufficiently large $N$,
\begin{equation}
 S_k'\leq C^kk\left(k^{\frac{2\nu}{\zeta}-1}\right)^k\leq C^{-N^\zeta}. \label{ffinal}
 \end{equation}

Going back to (\ref{summ}), by (\ref{sum'}) and (\ref{ffinal}) we have seen that
\begin{equation}
\sigma^N(t)\leq CN\left[e^{-\frac{\lambda}{2} N^q}e^{N^{\nu}C}+C^{-N^
\zeta}\right].
\end{equation}
Hence, by the choice 
$$
 q> \frac{45}{7} - \frac97 \alpha    \qquad \textnormal{and}  \qquad \frac{3\gamma+q}{2}<\nu<q
 $$
we have proven estimate (\ref{fff}).
\bigskip

\bigskip

\section{Generalizations}
\label{General}

Theorems \ref{th_02} and \ref{th_03} can be proven also with a weaker hypothesis on the initial data, that is substituting conditions (\ref{asp1})
and (\ref{asp})  with the following 
\begin{equation}
\int_{|i-x|\leq 1}dx \, g(|x|) \leq \frac{C}{|i|^\alpha}  \qquad  \textnormal{for any} \,\,\,\,   i\in {\mathbb{Z}}^3 \setminus \{0\}
\label{decay_den}
\end{equation}
being $\alpha$ in the respective intervals of Theorems \ref{th_02} and \ref{th_03}. This assumption has already been done in \cite{CCM16} and it
can be satisfied when  the spatial density $\rho(x,0)$
  has a suitable decay at infinity (as in (\ref{asp1}) and (\ref{asp})),  but also whenever $\rho(x,0)$ is constant, or has an oscillatory character, provided it has suitable support properties. 
Hence hypothesis (\ref{decay_den}) allows for spatial densities which possibly do not belong  to any $L^p$ space.

Moreover the initial data, both in positions and  velocities, are not necessarily intended to have spherical symmetry, 
provided that they satisfy the hypothesis of Theorems \ref{th_02} and \ref{th_03}.

 We want to  emphasize that when the magnetic lines are straight lines (condition that imposes an unbounded region 
$\Gamma$),  the previous results can be improved.
  In this article we have considered a region $\Gamma$ as a  torus,  and we had to deal with curved magnetic lines. This fact  complicated very much the problem. Actually a particle of velocity $v$ that remains close to the border of $\Gamma$ feels not only the electric force and the magnetic force (which does not increase its velocity),
 but also  a centrifugal force (proportional to $|v|^2$) much greater than the electric force, which affects considerably
 the dynamics. 
 In case the magnetic lines are straight lines this last force is absent and we can obtain stronger results.
 In particular we can extend the initial data to include a gaussian (Maxwell-Boltzmann) decay in the initial velocities.
 A region $\Gamma$ for which this analysis can be performed is for example an unbounded cylinder, 
 \begin{equation}
 \Gamma = \{x \in \R^3:   r<A  \} \qquad \qquad  r=\sqrt{x_2^2+x_3^2}
\end{equation}
 where the 
 magnetic field can be chosen as in \cite{inf},
\begin{equation}
 B(x)= (h(r^2),0,0) \qquad \qquad h(r^2) = \frac{1}{(A^2-r^2)^\tau} 
 \end{equation}
 or a half-space $\Gamma=\{ x\in \R^3: x_1>0\}$, with a magnetic field of the form
\begin{equation}
 B(x)= (0,0,h(x_1)) \qquad \qquad h(x_1) = \frac{1}{x_1^\tau}
 \end{equation} 
 and a suitable $\tau>0$.

 \bigskip
 
 \bigskip
 
 \bigskip


\noindent {\textbf{Acknowledgements}} \,\,
Work performed under the auspices of GNFM-INDAM and the Italian Ministry of the
University (MIUR).

 \bigskip
 
 \bigskip
 
 \bigskip

 \section*{Appendix}

\appendix

\setcounter{equation}{0}

\def\theequation{A.\arabic{equation}}

\medskip

\subsection*{Proof of Proposition \ref{media}}

From now on we will skip the index $N$ through the whole section. Moreover, we put for brevity 
$$
{\mathcal{V}}:={\mathcal{V}^N}(T)
$$
and 
$$
Q:=\sup_{t\in[0,T]}Q^N(R^N(t), t).
$$
Let  us define a time interval
\begin{equation}
\Delta:=
\frac{C}{{\cal{V}}^{2-\eta} \log \mathcal{V}}   \label{d1}
\end{equation}
with $\eta>0$ fixed later.

\medskip

 System  (\ref{chN}) in toroidal coordinates becomes
\begin{equation}
\label{tor_coor}
\left\{
\begin{aligned}
&-\dot\alpha^2 r + \ddot r - (R+r \cos\alpha) \dot\theta^2 \cos\alpha =
E_r + a'(r) \, \dot \theta   \\
&(R+r \cos\alpha) \ddot \theta + 2 \, \dot r \, \dot \theta   \cos\alpha 
-2 \, \dot\alpha \, \dot\theta \, r   \sin\alpha =
E_\theta  - \frac{a'(r) \, \dot r}{R+r \cos\alpha}    \\
&\ddot \alpha \, r + 2 \, \dot\alpha \, \dot r + (R+r \cos\alpha) \dot\theta^2 \sin\alpha  
= E_\alpha,  
\end{aligned}
\right.
\end{equation} 
where we denote by
$(E_r, E_\theta, E_\alpha)$ the components of the electric field in toroidal coordinates.
The components of the velocity in such coordinates are 
\begin{equation}
v_r = \dot r, \qquad  v_\theta = (R + r \cos\alpha) \dot \theta, \qquad v_\alpha = r \dot\alpha,
\label{tor_vel}
\end{equation}
and eq.s (\ref{tor_coor}) show that the magnetic field does not change the component $v_\alpha$ during the motion.

\noindent We consider two characteristics, solutions of (\ref{tor_coor}), 
\begin{equation}
\left( r_1(t), \theta_1(t), \alpha_1(t); \,\, \dot r_1(t), \dot\theta_1(t), \dot\alpha_1(t) \right)
\label{char_1}
\end{equation}

and
\begin{equation}
\left( r_2(t), \theta_2(t), \alpha_2(t); \,\, \dot r_2(t), \dot\theta_2(t), \dot\alpha_2(t) \right).
\label{char_2}
\end{equation}

\noindent  We give some preliminary Lemmas, proved in the following, in which we assume that the trajectories
remain in the region $r_0<r<\frac{R+r_0}{2}$.
This assumption is not essential, it is done only to avoid the singularity of the toroidal coordinates
for $r=R$. Actually for $r\geq \frac{R+r_0}{2}$ we could use cartesian coordinates, since the magnetic field is bounded,
and the analysis follows well known results as \cite{S}, \cite{W}.

\begin{lemma}\label{lem2}
 Let $t'\in [0,T]$ such that $t'+\Delta \in [0, T]$. The following holds: 
$$
\hbox{If} \qquad |r_1(t')  \dot\alpha_1(t')-r_2(t')\dot\alpha_2(t')|\leq  {\mathcal{V}}^{\eta}
$$
then
\begin{equation}
\sup_{t\in [t',  t'+\Delta]}|r_1(t)\dot\alpha_1(t)-r_2(t)\dot\alpha_2(t)|\leq 2  {\mathcal{V}}^{\eta}. \label{L1}
\end{equation}

\medskip

$$
\hbox{If} \qquad |r_1(t')\dot\alpha_1(t')-r_2(t')\dot\alpha_2(t')|\geq {\mathcal{V}}^{\eta}
$$
then
\begin{equation}
\inf_{t\in [t',t'+\Delta]}|r_1(t)\dot\alpha_1(t)-r_2(t)\dot\alpha_2(t)|\geq \frac12 {\mathcal{V}}^{\eta}. \label{L2}
\end{equation}
\end{lemma}

\bigskip

\noindent Let us put $v^{\perp}=\sqrt{v_r^2 + v_\theta^2}$, denoting the corresponding quantity
for the two characteristics as $v_i^{\perp}$,   $i=1, 2$.

\begin{lemma}\label{lem3}
 Let $t'\in [0,T]$ such that $t'+\Delta \in [0, T]$. The following holds:
  $$
\hbox{If} \qquad  |v_1^{\perp}(t')|\leq {\mathcal{V}}^\xi
$$
then
\begin{equation}
\sup_{t\in [t',t'+\Delta]}|v_1^{\perp}(t)|\leq2{\mathcal{V}}^\xi.   \label{L3}
\end{equation}

\medskip

$$
\hbox{If} \qquad   |v_1^{\perp}(t')|\geq {\mathcal{V}}^\xi
$$
then
\begin{equation}
\inf_{t\in [t',t'+\Delta]}|v_1^{\perp}(t)|\geq \frac{{\mathcal{V}}^\xi}{2},    \label{L4}
\end{equation}
\end{lemma}
where $\xi>0$ will be chosen later.

\bigskip

\begin{lemma}\label{lem4}
Let $t'\in [0,T]$  such that  $t'+\Delta \in [0, T]$,
and assume that 
$$
|v_{1,\alpha}(t')-v_{2,\alpha}(t')|\geq h {\mathcal{V}}^{\eta},    \quad \textit{for some} \,\, h\geq 1,
$$
$$
v_{1,\alpha} := r_1 \dot\alpha_1,  \qquad   v_{2,\alpha} := r_2 \dot\alpha_2.
$$
Then,  there exists $t_0\in [t',t'+\Delta]$ such that
$$
\left| X(t) - Y(t) \right|
\geq h \frac{{\mathcal{V}}^{\eta}}{8}|t-t_0|
$$
for all $t\in [t',t'+\Delta]$.
\end{lemma}

\bigskip

We resume now the proof of Proposition \ref{media},  in close analogy to what done in \cite{CCM16}.
We will use in the sequel cartesian coordinates for volume elements and integrand functions,
 and toroidal coordinates for the parametrization of the region of integration.

We fix any $t\in [0,T]$ such that $t+\Delta\leq T.$ For any $s\in [t, t+\Delta]$ we consider the time evolution of two characteristics $(X(s),V(s))$ and $(Y(s),W(s))$, both solutions to system (\ref{chN}), corresponding to  the characteristics (\ref{char_1}) and
 (\ref{char_2}) respectively. More precisely, we set 
\begin{equation*}
\begin{split}
&(X(s),V(s)):=(X(s,t,x,v),V(s,t,x,v)), \quad X(t)=x, \, V(t)=v\\& (Y(s),W(s)):=(Y(s,t,y,w),W(s,t,y,w)), \quad Y(t)=y, \, W(t)=w.
\end{split}\end{equation*}

\noindent Analogously we put
$$
v_{1, \alpha}(t_i)=v_{1, \alpha}, \qquad v_1^\perp(t_i)=v_1^\perp, \qquad 
v_{2, \alpha}(t_i)=v_{2, \alpha}, \qquad v_2^\perp(t_i)=v_2^\perp.
$$

\noindent Then, by the invariance of $f$ along the motion and the Liouville theorem we have,
\begin{equation}
\begin{split}
|E(X(s),s)|\leq & \int dydw \  \frac{f(y,w,s)}{|X(s)-y|^2}=
\int dydw \ \frac{ f(y,w,t)}{|X(s)-Y(s)|^2}.\label{Ei}
\end{split}
\end{equation}
We make a decomposition of the phase space $(y,w)\in \mathbb{R}^3\times \mathbb{R}^3$ in the following way:
\begin{equation}
T_1=\{(y,w): |y-x|\leq 2R(T)\}
\label{T1}
\end{equation}
\begin{equation}
S_1= \{ (y,w): |v_{1,\alpha}-v_{2,\alpha}|\leq {\cal{V}}^{\eta} \}\label{S1}
\end{equation}
\begin{equation}
S_2=\{ (y,w): \ |v_{2}^\perp|\leq {\cal{V}}^{\xi} \}\label{S2}
\end{equation}
\begin{equation}
S_3=(S_1 \cup S_2)^c
\end{equation}
where $\eta$ and $\xi$ are the positive parameters introduced in (\ref{d1}) and in Lemma \ref{lem3} respectively, to be suitably chosen in the following.

 We have
\begin{equation}
|E(X(s),s)|\leq\sum_{j=1}^4{\mathcal{I}}_j(X(s))\label{sum}
\end{equation}
where for any $s\in [t, t+\Delta]$
\begin{equation*}
{\mathcal{I}}_j(X(s))=\int 
_{T_1\cap S_j}dydw \  \frac{f(y,w,t)}{|X(s)-Y(s)|^2}, \quad \quad  j=1,2,3 
\end{equation*}
and 
\begin{equation*}
{\mathcal{I}}_4(X(s))=\int 
_{T_1^c}dydw \  \frac{f(y,w,t)}{|X(s)-Y(s)|^2}.
\end{equation*}

Let us start by the first integral. By the change of variables $(Y(s),W(s))=(\bar{y},\bar{w})$,
and Lemma \ref{lem2} we get
\begin{equation}
{\mathcal{I}}_1(X(s))\leq \int _{T_1'\cap S_1'}d\bar{y}d\bar{w}  \frac{f(\bar{y},\bar{w},s)}{|X(s)-\bar{y}|^2}
\end{equation}
where $T_1'=\{(\bar{y}, \bar{w}): |\bar{y}-X(s)|\leq 4R(T)\}$ and $S_1'= \{ (\bar{y},\bar{w}): |v_{1,\alpha}(s)-\bar{v}_{2,\alpha}|\leq 2 {\cal{V}}^{\eta} \}.$ 
 Now it is:
 \begin{equation}
\begin{split}
{\mathcal{I}}_1(X(s))\leq &\int _{T_1'\cap S_1'\cap \{|X(s)-\bar{y}|\leq \varepsilon\}}d\bar{y}d\bar{w} \frac{\  f(\bar{y},\bar{w},s)}{|X(s)-\bar{y}|^2}+\\&
\int _{T_1'\cap S_1'\cap \{  |X(s)-\bar{y}|>\varepsilon\}}d\bar{y}d\bar{w}\ \frac{f(\bar{y},\bar{w},s)}{|X(s)-\bar{y}|^2}\leq \\& C{\cal{V}}^{2+\eta} \varepsilon+ \int _{T_1'\cap S_1'\cap \{  |X(s)-\bar{y}|>\varepsilon\}}d\bar{y}d\bar{w}\ \frac{f(\bar{y},\bar{w},s)}{|X(s)-\bar{y}|^2}.\end{split}
\label{i11}\end{equation}
Now we give a bound on the spatial density $\rho(\bar{y},s).$ Setting 
$$
\rho_1(y,s)=\int_{S_1'} dw f(y,w,s),
$$
we have:
\begin{equation}
\begin{split}
&\rho_1(y,s)\leq  C {\cal{V}}^{\eta} \int_{|w^{\perp}|\leq a} dw^{\perp}+ \nonumber 
 \int_{|w^{\perp}|> a}dw^{\perp}\int dw_\alpha\ f(y,w,s) \leq   \nonumber \\
& C  a^2 {\cal{V}}^{\eta}+\frac{1}{a^2}\int dw |w|^2f(y,w,s)=C  a^2 {\cal{V}}^{\eta}+\frac{1}{a^2}K(y,s)
\nonumber
\end{split}
\end{equation}
where $K(y,s)=\int dw |w|^2f(y,w,s).$ Minimizing in $a$ we obtain
\begin{equation}
\rho_1(y,s)\leq C {\cal{V}}^{\frac{\eta}{2}} K(y,s)^{\frac12}.\label{K}
\end{equation}

Hence from (\ref{K}) we get
\begin{equation}
\begin{split}
\Bigg(\int _{T_1'}&dy\ \rho_1(y,s)^2\Bigg)^{\frac12}\leq \ C {\cal{V}}^{\frac{\eta}{2}} \Bigg(\int_{T_1'} dy\ K(y,s)\Bigg)^{\frac12}\leq \\ &C
{\cal{V}}^{\frac{\eta}{2}}\sqrt{W(X(s),4R(s),s)} \leq
C{\cal{V}}^{\frac{\eta}{2}}\sqrt{Q},
\end{split}
\label{rho1}
\end{equation}
where we have also applied Lemma \ref{lemR'/R}. Going back to (\ref{i11}), this bound implies
\begin{equation*}
\begin{split}
&{\mathcal{I}}_1(X(s))\leq\\&  C{\cal{V}}^{2+\eta} \varepsilon+\Bigg(\int_{T_1'} dy\rho_1(\bar{y},s)^2\Bigg)^{\frac12}\Bigg(\int_{T_1'\cap  \{|X(s)-\bar{y}|>\varepsilon\}} d\bar{y} \  \frac{1}{|X(s)-\bar{y}|^4}\Bigg)^{\frac12} \\&
\leq C\Big({\cal{V}}^{2+\eta}  \varepsilon+\sqrt{  \frac{   {\cal{V}}^{\eta}  Q }{\varepsilon}   }\Big).
\end{split}
\end{equation*}
By minimizing in $\varepsilon$ we obtain:
\begin{equation}
{\mathcal{I}}_1(X(s))\leq C Q^\frac13 {\cal{V}}^{\frac23(1+\eta)}.  \label{I1}
\end{equation}

\medskip

Now we pass to the term ${\mathcal{I}}_2.$ Proceeding as for the term ${\mathcal{I}}_1,$ defining $S_2'=\{ (\bar{y},
\bar{w}): \ |\bar{v}_2^{\perp}|\leq 2 {\mathcal{V}}^{\xi} \},$  by Lemma \ref{lem3} and the Holder inequality we get:
$$
{\mathcal{I}}_2(X(s))\leq  \int _{T_1'\cap S_2'}d\bar{y}d\bar{w} \frac{\  f(\bar{y},\bar{w},s)}{|X(s)-\bar{y}|^2}\leq 
$$
$$
 \int _{T_1'\cap S_2'\cap \{|X(s)-\bar y|\leq \varepsilon\}}d\bar{y}d\bar{w} \frac{\  f(\bar{y},\bar{w},s)}{|X(s)-\bar{y}|^2}+
 \int _{T_1'\cap \{  |X(s)-\bar{y}|>\varepsilon\}}d\bar{y} \frac{  \rho(\bar{y},s)}{|X(s)-\bar{y}|^2}\leq $$
$$
C {\mathcal{V}}^{1+2\xi}\varepsilon+
\left(\int_{T_1'}d\bar{y} \ \rho(\bar{y},s)^{\frac53}\right)^{\frac35}\left(\int_{ \{|X(s)-\bar{y}|> \varepsilon\}} d\bar{y} \  \frac{1}{|X(s)-\bar{y}|^5}\right)^{\frac25}.
$$
At this point we need a classical estimate, whose proof can be found e. g. in \cite{R3}.
For any $\mu\in \mathbb{R}$ and any positive number $R$ it is:
\begin{equation}
\int_{|\mu -x|\leq R}dx \, \rho^N(x,t)^{\frac53}\leq CW^N(\mu,R,t).\label{5/3}
\end{equation}
This, together with Proposition \ref{propo},  imply
$$
{\mathcal{I}}_2(X(s))\leq C{\mathcal{V}}(t)^{1+2\xi}\varepsilon+C Q^\frac35 \varepsilon^{-\frac45}.
$$
By minimizing in $\varepsilon$ we get:
\begin{equation}
{\mathcal{I}}_2(X(s))\leq C {\mathcal{V}}^{\frac49(1+2\xi)} Q^\frac13.
\label{I2}
\end{equation}

\medskip

Now we estimate ${\mathcal{I}}_3(X(s)).$ 

 We cover $ T_1\cap S_3$ by means of the sets
$A_{h,k}$ and  $B_{h,k}$, with 
${k=0,1,2,...,m}$ and ${h=1,2,...,m'}$, defined in the following way:
\begin{equation}
\begin{split}
A_{h,k}=\big\{ &(y,w): \ h {\cal{V}}^{\eta}< |v_{1,\alpha}-v_{2,\alpha}|\leq (h+1) {\cal{V}}^{\eta},\\& \alpha_{k+1}< |v_{2}^\perp|\leq \alpha_k, \  |X(s)-Y(s)|\leq l_{h,k}\big\}
\end{split}\label{Ak}
\end{equation}
\begin{equation}
\begin{split} 
B_{h,k}=\big\{ &(y,w):\ h {\cal{V}}^{\eta}< |v_{1 ,\alpha}-v_{2,\alpha}|\leq (h+1) {\cal{V}}^{\eta},\\& \alpha_{k+1}< |v_{2}^\perp|\leq \alpha_k, \  |X(s)-Y(s)|> l_{h,k}\big\}
\end{split}\label{Bk}
\end{equation}
where:
\begin{equation}
\alpha_k=\frac{{\cal{V}}}{2^k} \quad \quad l_{h,k}=\frac{2^{2k} }{h{\cal{V}}^{\beta+\eta} },
\label{al}
\end{equation}
with $\beta>0$ chosen later.
Hence we have
\begin{equation}
{\cal{I}}_3(X(s))\leq\sum_{h=1}^{m'} \sum_{k=0}^m\left({\cal{I}}_3'(h,k)+{\cal{I}}_3''(h,k)\right) \label{23}
\end{equation}
being
\begin{equation}
{\cal{I}}_3'(h,k)=\int_{T_1\cap A_{h,k}} \frac{f(y,w,t)}{|X(s)-Y(s)|^2} \, dydw
\end{equation}
and
\begin{equation}
{\cal{I}}_3''(h,k)=\int_{T_1\cap B_{h,k}} \frac{f(y,w,t)}{|X(s)-Y(s)|^2} \, dydw.
\end{equation}
Hence, recalling that we are in $S_3,$ it is
\begin{equation}
m\leq C \log {\cal{V}},  \qquad m'\leq C {\cal{V}}^{1-\eta}.\label{par}
\end{equation} 
By adapting Lemma \ref{lem2} and Lemma \ref{lem3} to this context it can be seen that $\forall \ (y,w)\in A_{h,k}$ it holds:
\begin{equation}
(h-1) {\cal{V}}^{\eta}\leq |v_{1,\alpha}(s)-v_{2,\alpha}(s)|\leq (h+2) {\cal{V}}^{\eta},   \label{lem31}
\end{equation}
and
\begin{equation}
\frac{\alpha_{k+1}}{2}\leq |v_{2}^\perp(s)|\leq 2\alpha_k.\label{lemperp1}
\end{equation}
Thus, putting
\begin{equation}\begin{split}
A_{h,k}'=&\big\{(\bar{y},\bar{w}):  \ (h-1){\cal{V}}^{\eta}\leq|v_{1,\alpha}(s)-\bar{v}_{2,\alpha}|\leq(h+2){\cal{V}}^{\eta}  \\&
\frac{\alpha_{k+1}}{2}\leq |\bar{v}_{2}^\perp|\leq 2\alpha_k,  
\ |X(s)-\bar{y}|\leq l_{h,k}  \big\},
\end{split}
\label{Akk}
\end{equation}
we have
\begin{equation}
{\cal{I}}_3' (h,k)\leq
\int_{T_1'\cap A_{h,k}'}\frac{f(\bar{y},\bar{w}, s)}{|X(s)-\bar{y}|^2} \, d\bar{y}d\bar{w}\label{int3}.
\end{equation}
The choice of the parameters $\alpha_k$ and $l_{h,k}$ made in (\ref{al}) implies that
\begin{equation}
\begin{split}
{\cal{I}}_3' (h,k) \leq &
 \ C \, l_{h,k}\int_{A_{h,k}'} \, d\bar{w}  \leq 
 \,C \, l_{h,k} \alpha_{k}^2  {\cal{V}}^{\eta} 
\leq C\frac{{\cal{V}}^{2-\beta}}{h} 
\end{split}
\end{equation}

from which, by (\ref{par}), it results
\begin{equation}
\sum_{h=1}^{m'} \sum_{k=0}^m{\cal{I}}_3'(h,k)\leq C {\cal{V}}^{2-\beta} \sum_{k=0}^m \sum_{h=1}^{m'}\frac{1}{h}\leq 
C {\cal{V}}^{2-\beta}  \log^2{\cal{V}}.
\label{i3}\end{equation}

Now we pass to ${\mathcal{I}}_3''(h,k).$
  Setting
 \begin{equation}
\begin{split}
C_{h,k}=\big\{& (y,w): \, h {\cal{V}}^{\eta}\leq |v_{1,\alpha}-v_{2,\alpha}|\leq (h+1) {\cal{V}}^{\eta},\\& \alpha_{k+1}\leq |v_{2}^\perp|\leq \alpha_k\big\},
\end{split}
\label{Bhk}
\end{equation}
 we have:
\begin{equation}
\begin{split}
&\int_{t}^{t+\Delta} {\mathcal{I}}_3''(h,k)\ ds\leq
\\&\int_{t}^{t+\Delta}ds\int_{T_1\cap C_{h,k}}dydw \ \frac{ f(y,w,t)}{|X(s)-Y(s)|^2}\, 
\chi(|X(s)-Y(s)|>l_{h,k})  \leq
\\&\int_{T_1\cap C_{h,k}} f(y, w, t) \left( \int_{t}^{t+\Delta} \frac{\chi(|X(s)-Y(s)|>l_{h,k})}{|X(s)-Y(s)|^2} \, ds \right)\, dy dw.
\end{split}
\label{ik}
\end{equation}
 Lemma \ref{lem4},
putting $a = \frac{4 \, l_{h,k}{\cal{V}}^{-\eta}}{h },$  gives us
\begin{equation}
\begin{split}
&\int_{t}^{t+\Delta} \frac{\chi(|X(s)-Y(s)|>l_{h,k})}{|X(s)-Y(s)|^2} \, ds 
\leq  \,\,   \\
&\int_{\{ s: |s-t_0|\leq a \}} \frac{\chi(|X(s)-Y(s)|>l_{h,k})}{|X(s)-Y(s)|^2} \, ds \, +\\
&\int_{\{ s: |s-t_0| > a \}} \frac{\chi(|X(s)-Y(s)|>l_{h,k})}{|X(s)-Y(s)|^2} \, ds   \leq  \\
&\frac{1}{l_{h,k}^2} \int_{\{ s: |s-t_0|\leq a \}} \, ds
+\frac{64 {\cal{V}}^{-2\eta}}{h^2} \int_{\{ s: |s-t_0| > a \}} \frac{1}{ |s-t_0|^2} \, ds \leq \\
& \quad \frac{2 a}{l_{h,k}^2} + \frac{C \,{\cal{V}}^{-2\eta}}{h^2 } \int_a^{+\infty} \frac{1}{s^2} \, ds
= \frac{C \,{\cal{V}}^{-\eta}}{ l_{h,k}h }.
\end{split}
\label{eq1}\end{equation}
Furthermore we have   
\begin{equation}
\begin{split}
\int_{T_1\cap C_{h,k}} f(y, w, t)\, dydw&\leq \frac{C}{\alpha_k^2}\int_{T_1\cap C_{h,k}} w^2 f(y, w, t) \, dydw, \\& \end{split}
\label{eq2}
\end{equation}
so that, by Lemma \ref{lemR'/R}
\begin{equation}
\begin{split}
\sum_{h=1}^{m'} \sum_{k=0}^m &\int_{T_1\cap C_{h,k}} w^2 f(y, w,t) \, dydw
\leq C \int_{T_1} K(y,t) \, dy \leq  \\
& CW(X(t), 3R(t), t) \leq C Q
\end{split}
\label{i5}
\end{equation}
where $K$ is the function introduced in (\ref{K}).

Taking into account (\ref{al}), by (\ref{ik}), (\ref{eq1}), (\ref{eq2}) and (\ref{i5}) we get:
\begin{equation}
\sum_{h=1}^{m'} \sum_{k=0}^m\int_{t}^{t+\Delta} {\mathcal{I}}_3''(h,k)\, ds\leq  
C  Q  {\cal{V}}^{\beta-2}. 
\label{wer} 
\end{equation}
By multiplying and dividing by $\Delta$ defined in (\ref{d1}) we obtain, 
\begin{equation}
\sum_{h=1}^{m'} \sum_{k=0}^m\int_{t}^{t+\Delta} {\mathcal{I}}_3''(h,k)\, ds \leq \, 
C  Q  {\cal{V}}^{\beta-\eta} \log \mathcal{V} \, \Delta.
\label{i3''}
\end{equation}
We minimize the exponents of $\mathcal{V}$ appearing in (\ref{I1}), (\ref{I2}),  (\ref{i3}) and (\ref{i3''}).
Recalling that $R^N(t)\leq C(1+\mathcal{V}^N(t)),$ by Proposition \ref{propo} and  Lemma \ref{Qbound}
 we obtain
\begin{equation}
\sum_{j=1}^3\int_{t}^{t+\Delta} {\mathcal{I}}_j(X(s))\, ds\leq 
 C \Delta {\mathcal{V}}^{\frac{15}{7}-\frac37 \alpha} \log^2{\mathcal{V}}.
 \label{ave}
\end{equation}
From here we can fix the previously introduced parameters as
\begin{equation}
\beta=\frac37\alpha-\frac17, \qquad \eta=\frac57 -\frac17\alpha, \qquad \xi=\frac{11}{14}-\frac{3}{28}\alpha.
\label{params}
\end{equation}

The last term ${\mathcal{I}}_4(X(s))$ can be bounded by the same procedure we used in  \cite{CCM16},
to conclude that it is bounded by a constant.

Hence by (\ref{Ei}), (\ref{sum}) and (\ref{ave}) we have
\begin{equation*}
 \int_{t}^{t+\Delta} |E(X(s),s)| \, ds\leq
C \Delta {\mathcal{V}}^{\frac{15}{7}-\frac37 \alpha} \log^2{\mathcal{V}} 
\end{equation*}
and then Proposition \ref{media}.
\qed

\subsection*{Proof of Lemma \ref{lem2}}

We consider the third equation of (\ref{tor_coor}), and we note that along this direction $\alpha$ the magnetic field
does  not act.
We get
\begin{equation}
\frac{d}{dt}(\dot\alpha \, r)= - \dot\alpha \, \dot r - \frac{v_\theta^2}{R+r\cos\alpha} \sin\alpha + E_\alpha,
\label{alpha_eq}
\end{equation}
and integrating in time,
\begin{equation}
v_\alpha(t) - v_\alpha(t') = \int_{t'}^t \left(- \dot\alpha(s) \, \dot r(s) - \frac{v_\theta^2(s)}{R+r(s)\cos\alpha(s)} 
\sin\alpha(s) + E_\alpha(s) \right) ds.
\label{a_eq2}
\end{equation}
The following bounds are fulfilled:  $|\dot r| = |v_r| \leq \mathcal{V}$,  $|v_\theta| \leq \mathcal{V}$,  $|E_\alpha| \leq C \mathcal{V}^{\frac43} Q^\frac13$ (recalling Proposition \ref{prop2}).  
For the term $\dot\alpha = \frac1r v_\alpha$, since it is $r>r_0$ (we are outside the torus)
we have $|\dot\alpha|\leq C \mathcal{V}$.
Therefore
\begin{equation}
|v_\alpha(t)| \leq |v_\alpha(t')| + C ( \mathcal{V}^2 + \mathcal{V}^{\frac43+\frac{3-\alpha}{3}})(t-t'),
\end{equation}
and  we get, for any $t\in[t',t'+\Delta]$,
$$
|r_1(t)\dot\alpha_1(t)-r_2(t)\dot\alpha_2(t)| = |v_{1,\alpha}(t)-v_{2,\alpha}(t)|\leq 
$$
$$
|v_{1,\alpha}(t')-v_{2,\alpha}(t')|+ C 
(\mathcal{V}^2 +\mathcal{V}^{\frac73-\frac{\alpha}{3}}) \Delta \leq 
$$
$$
\mathcal{V}^{\eta} + C \frac{\mathcal{V}^\eta}{\log \mathcal{V}} \leq 2 \mathcal{V}^{\eta}
$$
by the bound on $\alpha\in(\frac{8}{3}, 3)$, and being $\mathcal{V}$ sufficiently large.

The second statement can be proved in a similar manner:
$$
|v_{1,\alpha}(t)-v_{2,\alpha}(t)|\geq |v_{1,\alpha}(t')-v_{2,\alpha}(t')|-
C (\mathcal{V}^2 +\mathcal{V}^{\frac73-\frac{\alpha}{3}}) \Delta \geq
$$
$$
\mathcal{V}^{\eta} - C \frac{\mathcal{V}^\eta}{\log \mathcal{V}} \geq \frac12 \mathcal{V}^{\eta}
$$
for $\mathcal{V}$ sufficiently large.

\qed

\bigskip

\subsection*{Proof of Lemma \ref{lem3}}

The definition of $B$ and eq. (\ref{baa}) imply that
 \begin{equation}
\begin{split}
\frac{d}{dt}\left[v^{\perp}_1(t)\right]^2=2 v^{\perp}_1(t)\cdot E_{\perp}(t), \label{perp}\end{split}
\end{equation}
with $E_{\perp} = \sqrt{E_r^2 + E_\theta^2}$.
The thesis can be proved by contradiction. Assume that there exists a time interval $[t_1,t_2]\subset [t',t'+\Delta)$ such that 
$|v_1^{\perp}(t_1)|= \mathcal{V}^\xi$,   $|v_1^{\perp}(t_2)|= 2\mathcal{V}^\xi$  and  
$\mathcal{V}^\xi< |v_1^{\perp}(t)|< 2\mathcal{V}^\xi \ \ \forall t\in (t_1,t_2).$ 
Then from (\ref{perp})  it follows, by Proposition \ref{prop2} and (\ref{d1}):
\begin{equation}
\begin{split}
|v_1^{\perp}(t_2)|^2\leq \ |v_1^{\perp}(t_1)|^2&+2\int_{t_1}^{t_2} \ ds \ |v_1^{\perp}(s)|\,
|E_\perp(s)| \leq\\&
\mathcal{V}^{2\xi}+4\mathcal{V}^\xi\int_{t_1}^{t_2} ds \ | E(s)|\leq \\
&\mathcal{V}^{2\xi}+4 \mathcal{V}^{\xi} \Delta \, C \, \mathcal{V}^{\frac43} Q^\frac13\leq  \\
&\mathcal{V}^{2\xi} + C \frac{\mathcal{V}^{\xi+\frac13-\frac{\alpha}{3}+\eta}}{\log \mathcal{V}}<2 \mathcal{V}^{2\xi}.
\end{split}\label{A1}
\end{equation}
Hence we get an absurd, which proves the thesis. We have used in the last inequality  $\xi \geq \frac13 -\frac{\alpha}{3}+\eta$,
which holds by (\ref{params}) and since $\alpha \in (\frac83, 3)$. 

We pass now to prove (\ref{L4}).  We repeat the argument: assume that there exists a time interval $[t_1,t_2]\subset [t',t'+\Delta)$ such that 
$|v_1^{\perp}(t_1)|= \mathcal{V}^\xi$,    $|v_1^{\perp}(t_2)|= \frac{\mathcal{V}^\xi}{2}$  and  $\frac{\mathcal{V}^\xi}{2}< |v_1^{\perp}(t)|<  \mathcal{V}^\xi
\ \ \forall t\in (t_1,t_2).$ Then from (\ref{perp})
it follows, by Proposition \ref{prop2} and (\ref{d1}):
\begin{equation}
\begin{split}
|v_1^{\perp}(t_2)|^2\geq &\ |v_1^{\perp}(t_1)|^2-2\int_{t_1}^{t_2} \ ds \ |v_1^{\perp}(s)| \,
|E_\perp(s)| \geq\\&
\mathcal{V}^{2\xi}-2\mathcal{V}^\xi\int_{t_1}^{t_2} ds \ | E(s)|\geq \\
&\mathcal{V}^{2\xi}- C \frac{\mathcal{V}^{\xi+\frac13-\frac{\alpha}{3}+\eta}}{\log \mathcal{V}}>\frac{\mathcal{V}^{2\xi}}{2}.
\end{split}
\end{equation}
In this case again the contradiction proves the thesis.

\qed

\bigskip

\subsection*{Proof of Lemma \ref{lem4}}

 Defining
$$
\Omega(t) =  \int_{t'}^t \left[ v_{1,\alpha}(s) - v_{2,\alpha}(s) \right]   ds  + \lambda(t'),
$$
with
$$
\lambda(t')=r_1(t')\alpha_1(t') - r_2(t')\alpha_2(t'),
$$
let   $t_0\in [t',t'+\Delta]$ be the time at which 
$
\left| \Omega(t) \right|
$
has the minimum value.

\noindent Let us introduce also the function
$$
\bar{\Omega}(t)=\Omega(t_0)+ \dot{\Omega}(t_0)(t-t_0).
$$
By the choice of the magnetic field, it does not act on the  $\alpha$-component (in toroidal coordinates) of the velocity,
hence by (\ref{alpha_eq}) it is
\begin{equation}
\begin{split}
\frac{d^2}{d t^2}\left(\Omega(t)-\bar{\Omega}(t) \right)= - &\dot\alpha_1(t) \, \dot r_1(t) 
- \frac{[v_{1,\theta}(t)]^2 \sin\alpha_1(t)}{R+r_1(t)\cos\alpha_1(t)}  + E_\alpha(X(t),t) \, +  \\
&\dot\alpha_2(t) \, \dot r_2(t) 
+ \frac{[v_{2,\theta}(t)]^2 \sin\alpha_2(t)}{R+r_2(t)\cos\alpha_2(t)}  - E_\alpha(Y(t),t),     \\
&\Omega(t_0)=\bar{\Omega}(t_0), \quad  \dot{\Omega}(t_0)=\dot{\bar{\Omega}}(t_0),
\label{gam}
\end{split}
\end{equation}
introducing the notation $E_\alpha(X(t),t)$ for the $\alpha$-component of the electric field acting on the characteristic $1$,
and $E_\alpha(Y(t),t)$  for the $\alpha$-component of the electric field acting on the characteristic $2$.
Then we obtain
\begin{equation}
\begin{split}
\Omega(t)=\bar{\Omega}(t) \,&+  \nonumber \\
\int_{t_0}^tds\int_{t_0}^s d\tau\, &\Big[  - \dot\alpha_1(\tau) \, \dot r_1(\tau) 
- \frac{[v_{1,\theta}(\tau)]^2 \sin\alpha_1(\tau)}{R+r_1(\tau)\cos\alpha_1(\tau)}  + E_\alpha(X(\tau),\tau)  \nonumber \\
&+\dot\alpha_2(\tau) \, \dot r_2(\tau) 
+ \frac{[v_{2,\theta}(\tau)]^2 \sin\alpha_2(\tau)}{R+r_2(\tau)\cos\alpha_2(\tau)}  - E_\alpha(Y(\tau),\tau)  \nonumber   \Big]
\end{split}
\end{equation}
and proceeding in the same way as after (\ref{a_eq2})  (being  $r > r_0$ outside the torus),
\begin{equation}
\begin{split}
\int_{t_0}^tds\int_{t_0}^s d\tau\, &\left|  - \dot\alpha_1(\tau) \, \dot r_1(\tau) 
- \frac{[v_{1,\theta}(\tau)]^2 \sin\alpha_1(\tau)}{R+r_1(\tau)\cos\alpha_1(\tau)}  + E_\alpha(X(\tau),\tau) \right. \nonumber \\
&+\left. \dot\alpha_2(\tau) \, \dot r_2(\tau) 
+ \frac{[v_{2,\theta}(\tau)]^2 \sin\alpha_2(\tau)}{R+r_2(\tau)\cos\alpha_2(\tau)}  - E_\alpha(Y(\tau),\tau)  \nonumber \right| \leq
\nonumber  \\
& C \mathcal{V}^2  |t-t_0|^2 \leq C \mathcal{V}^2  |t-t_0| \Delta   \leq C \frac{\mathcal{V}^\eta}{\log \mathcal{V}} |t-t_0|.  \nonumber
\end{split}
\end{equation}
 Therefore
\begin{equation}
 |\Omega(t)|\geq |\bar{\Omega}(t)|-\frac{ \mathcal{V}^{\eta}}{4}|t-t_0|.
\label{z}
\end{equation}
It results:
$$
|\bar{\Omega}(t)|^2=|\Omega(t_0)|^2+2\Omega(t_0)\cdot \dot{\Omega}(t_0)(t-t_0)+|\dot{\Omega}(t_0)|^2 |t-t_0|^2,
$$
and we notice that that $\Omega(t_0) \cdot \dot{\Omega}(t_0) (t-t_0)\geq 0.$ In fact, if $t_0 \in(t',t'+\Delta)$ then $\dot{\Omega}(t_0)=0$ while if $t_0=t'$ or $t_0=t'+\Delta$ the product $\Omega(t_0)\cdot \dot{\Omega}(t_0) (t-t_0)\geq 0$.
Hence
$$
|\bar{\Omega}(t)|^2\geq |\dot{\Omega}(t_0)|^2 |t-t_0|^2.
$$
By Lemma \ref{lem2}  (adapted to this context with a factor $h\geq 1$), since $t_0\in [t',t'+\Delta]$ it is
$$
|\dot{\Omega}(t_0)|\geq h \frac{\mathcal{V}^{\eta}}{2},
$$
therefore
$$
|\bar{\Omega}(t)|\geq h \frac{\mathcal{V}^{\eta}}{2}|t-t_0|,
$$
and  by (\ref{z}) we get
\begin{equation}
 |\Omega(t)|\geq h \frac{\mathcal{V}^{\eta}}{4}|t-t_0|.
 \label{gam_t}
\end{equation}
Finally, we  achieve  Lemma \ref{lem4} by the following bound 
\begin{equation}
\left| \int_{t'}^t \left[ v_{1,\alpha}(s) - v_{2,\alpha}(s) \right] ds  + \lambda(t') \right| \leq 2 |X(t)-Y(t)|
\label{cord}
\end{equation}
which can be explained as follows:  the left hand side of (\ref{cord}) is the separation along the $\alpha$-coordinate (length of arc, if it is identically 
$r_1\equiv r_2$),
and in the worst case ($r_1\equiv r_2$), since we are  looking at  small  lengths, 
the double of the chord is greater than the length of the corresponding arc, for angles smaller than $\pi$.

\qed

\end{document}